\documentclass[Numbered, namedate,webpdf,traditional,small]{oup-authoring-template}

\usepackage{arydshln, amsmath}

\theoremstyle{thmstyleone}%
\newtheorem{theorem}{Theorem}
\newtheorem{proposition}[theorem]{Proposition}%
\newtheorem{example}{Example}%

\newtheorem{lemma}{Lemma}
\newtheorem{step}{Step}
\newtheorem{corollary}{Corollary}

\def\T{{ \mathrm{\scriptscriptstyle T} }}

\usepackage{xcolor}

\begin{document}

\journaltitle{Accepted by Science China Mathematics}
\DOI{https://doi.org/10.1360/SCM-2024-0039}
\copyrightyear{}

\firstpage{1}


\title[Optimal QS Design]{Optimal design of experiments with quantitative-sequence factors}

\author[1]{Yaping Wang}
\author[2]{Sixu Liu}
\author[3,$\ast$]{Qian Xiao\ORCID{0000-0001-7869-7109}}

\authormark{Wang, Liu and Xiao}

\address[1]{\orgdiv{KLATASDS-MOE, School of Statistics}, \orgname{East China Normal University} 
} 

\address[2]{\orgname{Beijing Institute of Mathematical Sciences and Applications}
}

\address[3]{\orgdiv{Department of Statistics, School of Mathematical Sciences}, \orgname{Shanghai Jiao Tong University}}

\address[*]{\orgdiv{Corresponding author: qian.xiao@sjtu.edu.cn}, \orgname{800 Dongchuan RD, Minhang District,
Shanghai, China, 200240}}





\abstract{
A new type of experiment with joint considerations of quantitative and sequence factors is recently drawing much attention in medical science, bio-engineering, and many other disciplines. The input spaces of such experiments are semi-discrete and often very large. Thus, efficient and economical experimental designs are required. Based on the transformations and aggregations of good lattice point sets, we construct a new class of optimal quantitative-sequence (QS) designs that are marginally coupled, pair-balanced, space-filling, and asymptotically orthogonal. The proposed QS designs have a certain flexibility in run and factor sizes and are especially appealing for high-dimensional cases.
}

\maketitle

\section{Introduction}
\label{intro}

In many modern scientific areas, a new type of experiment considering both the quantities and sequences to organize components, denoted as quantitative-sequence factors, becomes increasingly important \citep{wang2020, bh2015, wang2015bio, jourdain2009mixed}. As an illustration, \cite{wang2020} reported pioneering work on combinatorial drug therapy for lymphoma treatment where three FDA-approved drugs, denoted $c_1$, $c_2$, and $c_3$, were considered. They found that both doses of the drugs and their orders of addition had significant impacts on the efficacy of endpoint treatment. Three runs of this experiment are shown in Table~\ref{tab1}, where the dose units for drugs $c_1$, $c_2$, and $c_3$ are $\mu$ M, nM, and $\mu$ M, respectively. In each run, the three drugs were added one by one every 6 hours and the response (cell inhibition) was measured 6 hours after the addition of the last drug. Clearly, if we want to identify the optimal drug combination, both the doses of the drugs and their orders of addition need to be optimized simultaneously.
For another example in the production of nanocellulose (NC) gels \citep{bh2015}, a pretreatment process involves swelling agents, different acids, and enzymes to release hemicellulose. The orders in which the factors in the pretreatment are added along with their quantities are to be optimized for the NC size. In simulation experiments considering single-machine scheduling \citep{allahverdi1999review}, several jobs are to be processed on a single machine, where the time required to complete each job and the orders to process these jobs will affect the output. Such quantitative-sequence factors are also used in other physical and computer experiments \citep{shinohara1998, wang2015bio, jourdain2009mixed, panwalkar1973sequencing}.
\vspace{-0.1in}
\begin{table}[htbp]
\begin{minipage}{\linewidth}
\caption{Three runs of the drug combination experiment in \cite{wang2020}.
\label{tab1}}%
\begin{center}
\begin{tabular}{cccccccc}\hline
\multicolumn{1}{c}{Run} & \multicolumn{2}{c}{Drug $c_1$} & \multicolumn{2}{c}{Drug $c_2$} & \multicolumn{2}{c}{Drug $c_3$} & \multicolumn{1}{c}{Response} \\
\hline
& dose & order & dose & order & dose & order & \\
\hline
 1 & 3.75 & 1 & 95  & 2 & 0.16  & 3  & 39.91\\
 2 & 2.80 & 1 & 70  & 2 & 0.16  & 3  & 30.00\\
 3 & 3.75 & 3 & 95  & 1 & 0.16  & 2  & 34.68\\ \hline
\end{tabular}
\end{center}
\end{minipage}
\vspace{-0.1in}
\end{table}

For experiments with both quantitative and sequence input, practitioners often enumerate all possible sequences and apply factorial designs on quantities for each sequence \citep{wang2020}. However, this strategy requires as many as $m!s^m$ runs to arrange the $m$ components each having $s$ levels in its quantity, which can be prohibitively expensive for large $m$. Therefore, efficient and economic designs are required for this new type of experiments.
In the current experimental design literature, most researchers focused on either quantitative factors \citep{joseph2016space,wu2011experiments,luna2022orthogonal} or sequence factors, ie, order-of-addition factors \citep{van1995design, peng2017, voelkel2019design, lin2019order, schoen2023order, stokes2023metaheuristic}.
\cite{xiao2021}, \cite{yang2023ordering} and \cite{Tsai2023} considered experiments with both quantitative and sequence factors. The work of \cite{xiao2021} focuses mainly on modeling and active learning, and our paper focuses on experimental design.
{\cite{yang2023ordering} focused on $D$-optimal fractional factorial designs under a linear model.} 
\cite{Tsai2023} proposed a new class of design called dual-orthogonal arrays, which combine order-of-addition orthogonal arrays and two-level orthogonal arrays together. However, these designs allow only two-level quantitative factors to be investigated, which are rarely used in computer experiments.

Throughout the paper, we consider an $n$-run one-shot experiment for arranging $m$ components which are denoted as $c_1, \ldots, c_m$. Its experimental design $D=(X,O)$ includes runs $w_i = (x_i, o_i)$ for $i=1,\ldots,n$, where $x_i = (x_{i1}, \ldots, x_{im})$ includes the quantities of components and $o_i = (o_{i1}, \ldots, o_{im})$ is the sequence of arranging components.
As an illustration, if we add $1$ mM drug $c_2$ first, $2$ mM drug $c_3$ second, and $3$ mM drug $c_1$ last, we have $x = (1,2,3)$ and $o = (c_2, c_3, c_1) = (2,3,1)$.
In this paper, we propose novel algebraic constructions for optimal quantitative-sequence (QS) designs $D=(X,O)$ with run size $n$ can be multiples of $m$.
The resulting designs are proved to be marginally coupled, pair-balanced, space-filling, and asymptotically orthogonal. According to practical needs, they can fit both Gaussian process models and linear models in either one-shot or sequential experiments (active learning).

The remainder of the paper is organized as follows. In Section \ref{sec2},  we discuss the optimality criteria for the quantitative design $X$, the sequence design $O$, and their joint design $D=(X,O)$. In Section~\ref{sec3}, we propose novel construction methods for optimal QS designs with $n=m$ and show their theoretical properties.
In Section~\ref{sec4}, we discuss the general case of $n=km$ ($k \ge 2$).
Section \ref{sec5} presents a simulation case of the traveling salesman problem, demonstrating the application scenario of the constructed design and its advantages over other approaches.
Section \ref{sec6} concludes and discusses some future work. Appendix A provides a catalog of the QS designs obtained with small run-sizes and two algorithms used to find the optimal designs.
Appendix B includes all the proofs.


\section{Optimality criteria for QS designs D=(X,O)}\label{sec2}

The existing literature on modeling order-of-addition designs primarily employs linear models, including the pairwise ordering (PWO) model \citep{van1995design, voelkel2019design} and the component-position (CP) model \citep{yang2018,yang2023ordering}, among others. \cite{xiao2021} proposed a mapping-based additive Gaussian process (MaGP) model, which demonstrates strong predictive and optimization performance. In this work, we primarily investigate the optimal design properties under the MaGP model. Designs possessing these optimal properties can serve as initial designs for the active learning framework proposed in \cite{xiao2021}. Moreover, these optimal properties also enhance the estimation and prediction performance of the PWO and CP models as discussed in the following parts. The criterion used in this study was also adopted by \cite{xiao2021} and justified from the perspective of the MaGP model.

\subsection{Criteria for the quantitative design X}
\label{crix}
Space-filling Latin hypercube designs are desirable for quantitative factors \citep{LinTang2015, lin2016general, xiao2017construction, Santner2018}. In design $D=(X,O)$, we require that the quantitative part $X$ be a Latin hypercube design (LHD) and consider its space-filling property under the popular maximin distance criterion \citep{Johnson1990}.

Throughout the paper, denote $\mathcal{Z}^+_n =  \left\{1,\ldots,n\right\}$ and $\mathcal{Z}_n =  \left\{0,\ldots,n-1\right\}$ for any positive integer $n$.
Let $X(n, m)$ be an $n$-run LHD with $m$ factors where each column is a permutation of all elements in $\mathcal{Z}_{n}^+$ or $\mathcal{Z}_n$.
The maximin distance criterion seeks to scatter design points to fill the experimental domain such that the minimum distance between points is maximized. Define the $L_q$-distance between two runs $x_i$ and $x_j$ as
$d_q(x_i, x_j) =  \left\{ \sum_{k=1}^{m} \vert x_{ik}-x_{jk}\vert ^q \right\}^{1/q}
$ where $q$ is a positive integer. Define the $L_q$-distance of the design $X$ as $d_q(X) = \text{min} \{d_q(x_i, x_j),  1 \leq i<j \leq n \}$.
A design $X$ is called a maximin $L_q$-distance design if it has the largest $d_q(X)$ value. In this paper, we consider the popular Manhattan ($q=1$) and Euclidean ($q=2$) distances. For any $n\times m$ LHD, by \cite{zhou2015space}, the average row pairwise $L_1$-distance is $(n+1)m/3$, and the average squared row pairwise $L_2$-distance is $n(n+1)m/6$.  Since the minimum distance cannot exceed the integer part of the average, we have the following Lemma \ref{dupper}.

\begin{lemma}\label{dupper}
	For an $n$-run and $m$-factor LHD $X(n,m)$, we have
	$$
	d_1(X) \leq d_{1,upper} = \lfloor (n+1)m /3\rfloor,
        $$
        $$
    d_2(X) \leq d_{2,upper} = \sqrt{ \lfloor  n(n+1)m/6 \rfloor },
        $$
	where $\lfloor x \rfloor$ is the largest integer not exceeding $x$.
\end{lemma}

\subsection{Criteria for the sequence design O}
\label{crio}

A good sequence design $O$ should be pair-balanced \citep{xiao2021}. Since one key reason explaining why sequences matter is that components arranged in adjacent orders may have significant synergistic or antagonistic interactions, a good sequence design should test all such pairs of components.
Denote $t_{i,j}$ as the number of appearances of sub-sequence ``$i,j$'' in all rows of the design $O$ for $1\leq i \neq j \leq m$. A sequence design $O$ is called pair-balanced if each component is preceded by every other component the same number of times; that is, all possible $t_{i,j}$'s are equal. We illustrate this in the following Example~\ref{eg:o}.
The pair-balanced property is also studied in the literature on crossover design \citep{Bose2009}. However, the experiments that involve sequences in this paper are different from the crossover trials. Here, only the endpoint response is measured and the component effects are assumed to be dependent on their arrangement orders \citep{lin2019order, robert2018}.

\begin{table*}[htbp]
\caption{Comparison of designs' $t_{i,j}$ pairs.\label{method1}}
\tabcolsep=0pt
\begin{tabular*}{\textwidth}{@{\extracolsep{\fill}}ccccccccccccc@{\extracolsep{\fill}}}
\toprule%
			$t_{i,j}$ & $c_1c_2$ & $c_1c_3$ & $c_1c_4$ & $c_2c_1$ & $c_2c_3$ & $c_2c_4$ & $c_3c_1$ & $c_3c_2$ & $c_3c_4$ & $c_4c_1$ & $c_4c_2$ & $c_4c_3$ \\
			\midrule
			$O_1$ &   3 &   0 &   0 &    0 &   3 &   0 &   0 &   0 &   3 &   3 &   0 &   0 \\
			$O_2$ &   1 &   1 &   1 &   1 &   1 &   1 &   1 &   1 &   1 &   1 &   1 &   1 \\
			\botrule
\end{tabular*}
\end{table*}

\begin{example}
\label{eg:o}
Consider a drug combination experiment involving four drug components. Table \ref{method1} lists the $t_{i,j}$ pairs of two candidate designs $O_1$ and $O_2$. When arranging drug sequences, researchers need to know the possible synergistic or antagonistic pairwise interactions between drugs. For example, if Drug $c_1$ and $c_2$ have a strong synergistic interaction, they should be adjacent in the sequence; if they have strong antagonistic effects, they should be apart. Under this consideration, it is clear that the pair-balanced design {$O_2$} is better where each adjacent pair appears once (equal number of times) for testing the interactions.
$$
O_1 =
\begin{pmatrix}
  	c_1 & c_2 & c_3 & c_4 \\
	c_2 & c_3 & c_4 & c_1 \\
	c_3 & c_4 & c_1 & c_2 \\
	c_4 & c_1 & c_2 & c_3 \\
\end{pmatrix},
O_2 =
\begin{pmatrix}
  c_1 & c_2 & c_3 & c_4 \\
  c_2 & c_4 & c_1 & c_3 \\
  c_3 & c_1 & c_4 & c_2 \\
  c_4 & c_3 & c_2 & c_1 \\
\end{pmatrix}.
$$
\end{example}

In the current literature, researchers propose to use space-filling or orthogonal sequence designs, i.e., order-of-addition designs \citep{peng2017, voelkel2019design, wang2020, yang2018}. To gain more information from the data, space filling designs minimize similarities among runs, and orthogonal designs alleviate the associations among factors. Note that the sequence design $O$ here consists of qualitative inputs, for example $o = (c_2, c_3, c_1) = (2,3,1)$ in Section~\ref{intro}.
Thus, we adopt the maximin Hamming distance criterion to measure the space-filling property of design $O$. The Hamming distance counts the number of positions in two runs where the corresponding levels differ, which is widely used for qualitative levels. Denote the Hamming distance between the $i$th and $j$th runs in $O$ as $d_H(o_i, o_j) =  \sum_{k=1}^{m} I(o_{ik}, o_{jk})$ where the indicator function $I(o_{ik}, o_{jk}) = 1$ if $o_{ik} \neq o_{jk}$; otherwise 0. Let the Hamming-distance of design $O$ be $d_H(O) = \text{min} \{d_H(o_i, o_j) ,  1\leq i<j\leq n \}.$ The maximin Hamming distance design seeks to maximize $d_H(O)$ to achieve space-filling.

\begin{lemma}\label{hupper}
For an $n$-run sequence design $O$ arranging $m$ components,
\begin{equation*}
d_H(O) \leq d_{H,upper}=\left\{\begin{array}{ll}
m, & \ n\leq m, \\
m-1, & \ n > m.
\end{array}
\right.
\end{equation*}
\end{lemma}

Orthogonal designs are widely used in experimentation \citep{hedayat1999, chen2022study}. Orthogonality is typically classified into combinatorial orthogonality and column orthogonality. Combinatorial orthogonality generally requires the design to be a strength-2 orthogonal array, meaning that every pair of columns contains all level combinations with equal frequency. In contrast, column orthogonality only requires that the correlation coefficient between any two columns of the design is zero. Combinatorial orthogonality implies column orthogonality, but the reverse is not necessarily true. For a sequential design $O$, where each row represents a permutation of all components, the orthogonality defined by \cite{yang2018} is an extension of traditional combinatorial orthogonality. Specifically, it requires that every pair of columns $(i, j)$, where $i \neq j $, exhibits all level combinations with equal frequency. However, this type of orthogonality imposes a constraint that the number of rows $n$ must be a multiple of $m(m-1)$. Since this paper focuses on scenarios with a smaller number of rows (e.g. $n = m, 2m, \dots$), combinatorial orthogonality cannot be satisfied. Although a sequential design with good column orthogonality does not necessarily exhibit good combinatorial orthogonality, \cite{LinTang2015} suggests that a well-constructed space-filling design or an approximately combinatorially orthogonal design should be column-orthogonal or approximately column-orthogonal. This implies that good sequential designs can be easily selected from a set of column-orthogonal or approximately column-orthogonal designs. Moreover, the criterion for column orthogonality is computationally simple. Given these considerations, we adopt the column orthogonality for selecting $O$ and refer to it simply as "orthogonality" in this work for convenience.

Define the average absolute correlation of design $O$ as $r_{ave}(O)= \sum_{u\neq v} \vert r_{uv}(O) \vert / \{m(m-1)\}$ where $r_{uv}(O)$ is the correlation between the $u$th and $v$th columns in $O$ and $u,v=1,\ldots,m$. An orthogonal design has $r_{ave}(O)=0$. An asymptotically orthogonal design has $r_{ave}(O) \rightarrow 0$ as the run size $n \rightarrow \infty$. Generally speaking, designs with low values of $r_{ave}$ are preferred \citep{tang1997method, sun2011construction,wang2018optimal}. In general, optimal sequence designs $O$ should be pair-balanced, space-filling and orthogonal (or asymptotically orthogonal). The first two criteria among these three were proposed by \cite{xiao2021} and utilized in optimizing the initial design for the MaGP model. The orthogonality criterion also contributes to improving the space-filling properties and combinatorial orthogonality of $O$, thereby enhancing the modeling performance of MaGP, CP, and PWO.

\subsection{Criterion for joining designs X and O}

\cite{Deng2015} proposed the marginally coupled designs (MCDs) as a new class of optimal designs for experiments with both quantitative and qualitative factors. The MCD maintains an economic run size with many attractive properties, where the design points for the quantitative factors form a sliced Latin hypercube design \citep{qian2012sliced}. For the construction of MCDs, refer to \cite{Deng2015,he2017construction,he2019construction,yang2023doubly}.
Due to the semidiscrete nature of sequence factors, the marginally coupled structure is also desirable for QS designs. Note that the entries in the sequence design $O$ (i.e. the components to be arranged) are qualitative. A marginally coupled QS design $D=(X,O)$ requires that the quantitative part $X$ is an LHD and for each level of any factor in $O$ the corresponding design points in $X$ also form an LHD. This structure for $D$ with the special design size $n=m$ is straightforward, and in this paper we focus on the general cases of $n \geq 2m$.
Refer to Examples~\ref{ex0} and~\ref{ex1} in Section~\ref{sec4} for illustrations of the marginally coupled structure.

\section{Optimal QS design of run size \lowercase{$n=m$}}\label{sec3}

\subsection{A special case of m}
\label{constructnm}
We start from the case when the run size $n$ equals to the number of components $m$ for $m = p-1$ and $p$ being any odd prime. Such small QS designs are suitable as initial designs in sequential experiments (i.e., active learning) under Gaussian process models \citep{xiao2021}. Larger run sizes are required in one-shot experiments, which will be discussed in the next section.

A good lattice point set $D_0$ is a $p \times m$ matrix whose $i$th row is $h \times i \text{ (mod } p\text{)}$ where $i = 1, \ldots, p$ and vector $h=(1, \ldots, m)$ \citep{zhou2015space}. Let $W: \mathcal{Z}_p \rightarrow \mathcal{Z}_p$ be the Williams transformation \citep{williams1949, wang2018optimal}, where
\begin{equation}
\label{wt}
W(x)=\left\{\begin{array}{ll}
2x, & \ 0\leq x < p/2, \\
2(p-x)-1, & \ p/2\leq x \leq p-1.
\end{array}
\right.
\end{equation}
\cite{wang2018optimal} proposed the following three-step procedure to construct maximin $L_1$-distance LHDs for quantitative factors only.
\begin{step}
Generate a good lattice point set $D_0$.
\end{step}
\begin{step}
For any $b\in  \left\{0,\ldots,p-1\right\}$, generate designs $D_b=D_0+b \text{ (mod } p\text{)}$ and designs $E_b=W(D_b)$.
\end{step}
\begin{step}
Among the $p$ possible designs $E_b$,  find the best one under the maximin distance criterion.
\end{step}

The last rows of $D_b$ and $E_b$ are $(b, \ldots, b)$ and $(W(b), \ldots, W(b))$, respectively. Define the leave-one-out designs $\tilde D_b$ and $\tilde E_b$ by deleting the last rows of $D_b$ and $E_b$ and then reordering their levels so that $\tilde D_b$ and $\tilde{E}_b$ are $m \times m$ LHDs with levels in $\mathcal{Z}^+_{m} = \left\{1,\ldots,m\right\}$.

\begin{theorem}\label{lem:balance}
	Let $p$ be any odd prime and $n=m=p-1$. Let the sequence design $O$ be either the leave-one-out design $\tilde D_b$ or $\tilde  E_b$ for any $b\in \mathcal{Z}_p = \left\{0,\ldots,p-1\right\}$. Then, \\
  (i) $O$ is a Latin square whose rows and columns are permutations of $\left\{1,\ldots,m\right\}$; \\
  (ii) $O$ is a pair-balanced design, where all possible adjacent component pairs appear once, i.e. $t_{i,j}=1$ for $1\leq i \neq j \leq m$; \\
  (iii) $O$ is a maximin Hamming distance design with $d_H(O) = m$ which achieves the upper-bound in Lemma~\ref{hupper}.
\end{theorem}

Theorem \ref{lem:balance} guarantees the pair-balanced and space-filling properties of the sequence designs constructed by either setting $O = \tilde D_b$ or $O = \tilde  E_b$ for any $b\in \mathcal{Z}_p$. To further improve the sequence design $O$ in terms of orthogonality, we propose to choose the best design $O = \tilde  E_{b^*_1}$ which has the lowest $r_{ave}$ value among all possible $b\in \mathcal{Z}_p$; that is, we select
\begin{equation}
\label{b1_star}
b^*_1 = \arg \min_{b\in \mathcal{Z}_p} r_{ave} (\tilde  E_{b}).
\end{equation}

\begin{theorem}
\label{theo:XO_m}
	Let $p$ be any odd prime, and $n=m=p-1$. Let the sequence design $O=\tilde E_{b_1^*}$ with $b^*_1$ determined by \eqref{b1_star}. Let the quantitative design $X = \tilde E_{b^*_2}$ with $b^*_2 = W^{-1} ( (p-1)/2 \pm c)$ where $W^{-1}(\cdot)$ is the inverse of the Williams transformation defined in \eqref{wt} and $c = \lfloor (p^2-1)/12 \rfloor$ if $c^2 + 2(c+1)^2 \geq (p^2-1)/4$; otherwise, $c = \lfloor (p^2-1)/12 \rfloor + 1$.
	Then, the constructed QS design $D=(X,O)$ has the following properties in addition to those in Theorem~\ref{lem:balance}: \\
	(i) $O$ is a Hamming equidistant, pair-balanced and asymptotically orthogonal design with $$r_{ave}(O) < \mathcal{O}(1/m) \rightarrow 0 \text{  as  } m\rightarrow \infty;$$
	(ii) $X$ is an asymptotically maximin $L_1$-distance design which also has large $L_2$-distance:
  $$d_1(X) \geq (1-\mathcal{O}(1/m)) d_{1,upper} \rightarrow d_{1,upper} \text{ \  as  \ } m\rightarrow \infty,$$
  $$d_2(X) \geq  \sqrt{{2}/{3}} (1-\mathcal{O}(1/m))  d_{2,upper} \rightarrow 0.817d_{2,upper} \text{ \  as  \ } m\rightarrow \infty ,$$
  where $\mathcal{O}$ is the big O notation and the upper-bounds $d_{1,upper}$ and $d_{2,upper}$ are given in Lemma \ref{dupper}.
\end{theorem}

Theorem \ref{theo:XO_m} presents our design construction of $D=(X,O)$ with the run size $n=m=p-1$ and $p$ being any odd prime. Here, the scalar $b^*_1$ can be easily found by comparing $p$ candidate designs $\tilde  E_{b}$ with $b\in \mathcal{Z}_p$, and $b^*_2$ can be directly calculated.
In Table~\ref{tab:b_and_b1}, we list the values of $b^*_1$ and $b^*_2$ in the cases of $5\leq p < 100$, that is, $4\leq m < 99$.
From Table~\ref{tab:b_and_b1}, we can observe that there are at least two possible values for $b^*_1$, and the following Proposition \ref{prop:twob} sheds some light on this phenomenon.

\begin{table*}[htbp]
\caption{Lists of $b_1^*$ and $b^*_2$ values for some small odd prime $p$.\label{tab:b_and_b1}}
\tabcolsep=0pt
\begin{tabular*}{\textwidth}{@{\extracolsep{\fill}}ccc ccc ccc@{\extracolsep{\fill}}}
\toprule%
$p$  & $b^*_1$  & $b^*_2$   & $p$   & $b^*_1$  & $b^*_2$ & $p$  & $b^*_1$  & $b^*_2$  \\
\midrule
    5    &1,3,4 &  3,4  &  31    & 4,11  & 3,12  & 67    & 42,58 & 7,26   \\
    7     & 1,2 &  4,6  & 37    & 23,32 & 4,14   & 71    & 9,26  & 43,63  \\
   11     & 7,9 &  1,4  & 41    & 5,15  & 3,12   & 73    & 9,27   & 44,65 \\
   13    & 8,11 &  1,5  & 43    & 27,37 & 4,14   & 79    & 10,29  & 8,31 \\
   17    & 2,6  & 10,15  & 47    & 6,17  & 4,16   & 83    & 52,72  & 50,74 \\
   19    & 2,7  &  2,7   & 53    & 33,46 & 26,38  & 89    & 11,33  & 9,35 \\
   23    & 3,8  & 14,20  & 59    & 37,51 & 6,23   & 97    & 12,36  & 10,38 \\
   29    & 18,25 & 3,11  & 61    & 38,53 & 37,54  &       &        &     \\
   \botrule
\end{tabular*}
\end{table*}

\begin{proposition}\label{prop:twob}
For $b,b' \in \mathcal{Z}_p$, $b\neq b'$, $m=p-1$ and $p$ being any odd prime, if $b+b' = (p-1)/2$ or $b+b'=(3p-1)/2$, then
  $r_{ave} (\tilde  E_{b}) = r_{ave} (\tilde  E_{b'})$.
\end{proposition}

When $m=p-1$ where $p$ is an odd prime and $q=2m+1$ is also an odd prime, the quantitative design $X$ in Theorem~\ref{theo:XO_m} can be further improved under the maximin $L_1$-distance criterion. Let $w: \mathcal{Z}_q \rightarrow \mathcal{Z}_q$ be the modified Williams transformation defined by
$$
w(x)=\left\{\begin{array}{ll}
2x, & \ 0\leq x <q/2 \\
2(q-x), & \ q/2\leq x<q,
\end{array}
\right.
$$
Let $D_0$ be a $(2m+1)\times (2m)$ GLP design and $A_1$ be the $m\times m$ leading principal submatrix of $D_0$.
Following Theorem~4 of \cite{wang2018optimal}, $E = w(A_1)/2$ is an $L_1$-equidistant $LHD(m,m)$ with $d_1(E)=\bar{d}_1 = m(m+1)/3$, where $w(A_1)$ is the $m\times m$ design with all entries in $A_1$ transformed by the modified Williams transformation $w(\cdot)$.

\begin{corollary}
\label{theo:XO_m_cor}
	Let $p$ be an odd prime, $n=m=p-1$. If $2m+1$ is also an odd prime, substitute the quantitative design $X$ in Theorem~\ref{theo:XO_m} by the $L_1$-equidistant $LHD(m,m)$ $E = w(A_1)/2$ defined above. Then $X$ is an exact maximin $L_1$-distance design with $d_1(X) = d_{1,upper}$.
\end{corollary}

Consider all the cases of $n=m<100$.  When $p=7, 19, 31, 37, 79, 97$, i.e. $m=6, 18, 30, 36, 78, 96$, the condition of Corollary~\ref{theo:XO_m_cor} holds, hence we can take $X$ to be the $L_1$-equidistant LHDs in Corollary~\ref{theo:XO_m_cor}.
These designs can improve the quantitative design $X = \tilde E_{b_2^*}$ in Theorem~\ref{theo:XO_m} with $m=18, 30, 36, 78, 96$ which are not $L_1$-equidistant.
For $m=6$, both $X = \tilde E_{b_2^*}$ and $X=w(A_1)/2$ are $L_1$-equidistant, and we can choose
the one with the {maximum} $L_2$-distances. We illustrate the construction of $n=m=6$ in the following example.

\begin{example}
\label{ex:m=6}
Consider the construction of optimal QS design $D=(X,O)$ with $n=m=6$. First, we construct the leave-one-out Williams transformed good lattice point design $\tilde E_{b_1^*}$ to be the sequence part $O$. Given that $p=m+1=7$ and $h=(1,\ldots,6)$, we construct the good lattice point design $D_0$ whose $i$th row is $h \times i \text{ (mod } p\text{)}$ for $i = 1, \ldots, p$. We choose $b_1^* = 1$ by Table~\ref{tab:b_and_b1}. Then $D_{b_1^*} = D_1=D_0+1 \text{ (mod }7)$ and $E_1=W(D_1)$, where the Williams transformation for $p=7$ is shown in Table~\ref{wleg2}. After removing the last row of $ E_1$ and relabeling the remaining design, we have the required Latin square $O=\tilde E_1$.

$$
D_0 =
\begin{pmatrix}
  		1 & 2 & 3 & 4 & 5 & 6 \\
		2 & 4 & 6 & 1 & 3 & 5 \\
		3 & 6 & 2 & 5 & 1 & 4 \\
		4 & 1 & 5 & 2 & 6 & 3 \\
		5 & 3 & 1 & 6 & 4 & 2 \\
		6 & 5 & 4 & 3 & 2 & 1 \\
		0 & 0 & 0 & 0 & 0 & 0 \\
\end{pmatrix},
E_1 =
\begin{pmatrix}
  	4 & 6 & 5 & 3 & 1 & 0 \\
	6 & 3 & 0 & 4 & 5 & 1 \\
	5 & 0 & 6 & 1 & 4 & 3 \\
	3 & 4 & 1 & 6 & 0 & 5 \\
	1 & 5 & 4 & 0 & 3 & 6 \\
	0 & 1 & 3 & 5 & 6 & 4 \\
	2 & 2 & 2 & 2 & 2 & 2 \\
\end{pmatrix},
$$
$$
\tilde E_1 =
\begin{pmatrix}
  	4 & 6 & 5 & 3 & 2 & 1 \\
	6 & 3 & 1 & 4 & 5 & 2 \\
	5 & 1 & 6 & 2 & 4 & 3 \\
	3 & 4 & 2 & 6 & 1 & 5 \\
	2 & 5 & 4 & 1 & 3 & 6 \\
	1 & 2 & 3 & 5 & 6 & 4 \\
\end{pmatrix}.
$$

\vspace{-0.1in}
\begin{table}[htbp]
\begin{minipage}{\linewidth}
\caption{The Williams transformation for $p=7$.
\label{wleg2}}%
\begin{center}
\begin{tabular}{ccccccccc}
\hline
$x$ && 0 & 1 & 2 & 3 & 4 & 5 & 6 \\
$W(x)$ && 0 & 2 & 4 & 6 & 5 & 3 & 1\\
\hline
\end{tabular}
\end{center}
\end{minipage}
\end{table}
\vspace{-0.1in}

Next, similar to the first step, we can construct the quantitative part $X = \tilde E_{b_2^*}$ using the chosen $b_2^* = 4$ by Table~\ref{tab:b_and_b1}. Here $\tilde E_{b_2^*}$ is an LHD with $d_1(\tilde E_{b_2^*}) = 14$ and $d_2(\tilde E_{b_2^*}) = \sqrt{34}$. Note that the bounds in Lemma~\ref{dupper} are $d_{1,upper}=14$ and $d_{2,upper}=\sqrt{42}$, respectively. Thus, $\tilde E_{b_2^*}$ is a maximin $L_1$-equidistant design.

As $2m+1=13$ is an odd prime, we can also follow Corollary~\ref{theo:XO_m_cor} to construct another
$L_1$-equidistant design as the quantitative part, i.e., $X=E = w(A_1)/2$. We have $d_1(E)=d_1(\tilde E_{b_2^*})=14$  and $d_2(E) = \sqrt{40} > d_2(\tilde E_{b_2^*})$. Therefore, we choose $X=E$ as it performs better under the maximin $L_2$-distance criterion.

Finally, the combined design $D = (X,O)$ is shown in Table~\ref{tab:exam1}. The quantitative design $X$ is a maximin $L_1$-distance design and has large $L_2$-distance ($d_2(X)/d_{2,upper} = 0.976$). The sequence design $O$ has $d_H(O)=6$ (which equals to the $d_{H,upper}$ in Lemma~\ref{hupper}), $r_{ave}(O) = 0.2$, and all $t_{i,j}=1$ for $1\leq i\neq j \leq 6$.
\vspace{-0.1in}
\begin{table}[htbp]
\begin{minipage}{\linewidth}
\caption{The optimal QS design with $n=m=6$.
\label{tab:exam1}}%
\begin{center}
\begin{tabular}{cccccc c cccccc}
\multicolumn{6}{c}{$X$}  && \multicolumn{6}{c}{$O$}   \\
1 & 2 & 3 & 4 & 5 & 6 && 4 & 6 & 5 & 3 & 2 & 1 \\
2 & 4 & 6 & 5 & 3 & 1 && 6 & 3 & 1 & 4 & 5 & 2 \\
3 & 6 & 4 & 1 & 2 & 5 && 5 & 1 & 6 & 2 & 4 & 3 \\
4 & 5 & 1 & 3 & 6 & 2 && 3 & 4 & 2 & 6 & 1 & 5 \\
5 & 3 & 2 & 6 & 1 & 4 && 2 & 5 & 4 & 1 & 3 & 6 \\
6 & 1 & 5 & 2 & 4 & 3 && 1 & 2 & 3 & 5 & 6 & 4 \\
\end{tabular}
\end{center}
\end{minipage}
\end{table}
\end{example}
\vspace{-0.3in}

In Example~\ref{ex:m=6}, the design has $r_{ave}(O) = 0.2$ and looks a bit large for this small design case. However, as $m$ increases, the maximum absolute correlation decreases quickly. For example, when $m=58$, we have $r_{ave}(O) = 0.018$; see the right panel of Figure~\ref{fig1} for more information.

\cite{xiao2021} proposed a sequential experimentation (i.e. active learning) starting from a class of initial designs when the run sizes $n=m=p-1$ for $p$ are odd primes.
Compared to their designs, the proposed ones in this paper are more space-filling and orthogonal, as shown in Figure~\ref{fig1}. The left graph shows the $L_q$-distance ratios $d_q(X)/d_{q,upper}$ of the quantitative designs $X$, where the upper-bounds $d_{q,upper}$ are defined in Lemma~\ref{dupper} for $q=1,2$. The right graph shows the average absolute correlations $r_{ave}(O)$ of the sequence designs $O$. Clearly, our designs are more space-filling in $X$ and have smaller average absolute correlations in $O$. Moreover, as $m$ increases, the $L_1$-distance ratio approaches 100\% and $r_{ave}$ approaches $0$ in our design construction.

\begin{figure*}[htbp]
\centering
\includegraphics[angle=0, scale=0.5]{{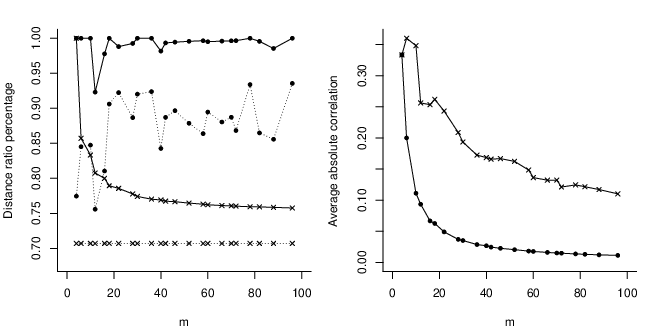}}\par
\caption{Left: $L_1$-distance (solid) and $L_2$-distance (dotted) ratios of quantitative designs $X$ generated by our method (circle) and the competitor's method (cross). Right: $r_{ave}$ values of sequence designs $O$ generated by our method (circle) and the competitor's method (cross).}\label{fig1}
\end{figure*}

\subsection{A more general case of m}\label{sec3.2}

In this subsection, we construct QS designs for some cases where $n=m$ and $m+1$ is not an odd prime number.
Our construction is based on two classes of Latin squares constructed by \cite{yin2022distance} and \cite{williams1949}, respectively.

The following method for constructing Latin squares is proposed by \cite{yin2022distance}.
Let $N>2$ be a positive integer, $m=\phi(N)/2$ and $h=(h_1,\ldots,h_m)$ be all integers less than $\lfloor N/2 \rfloor$ and co-prime to $N$, where $\phi(N)$ is the Euler function that counts the number of positive integers that are less than and coprime to $N$. It is easy to see that $\phi(N)$ is even and thus $m$ is an integer.
Let
\begin{equation}
\label{eq:L}
  L = (l_{ij})_{m\times m},
\end{equation}
where $ l_{ij} = \min\{ h_i \times h_j  \text{ (mod } N\text{)},   N-h_i \times h_j  \text{ (mod } N\text{)}\}$ for $i,j = 1, \ldots, m$.
It is straightforward to show that $L$ is an $m\times m$ Latin square whose rows and columns are permutations of $h_1,\ldots,h_m$.

Now we review another construction of the Latin square proposed by \cite{williams1949}.
Let $m$ be any even positive integer.
Consider the Williams transformation $W(\cdot)$ defined in \eqref{wt} from $\mathcal{Z}_m \rightarrow \mathcal{Z}_m$.
Let $h=(W^{-1}(0),W^{-1}(1),\ldots, W^{-1}(m-1))$ where $W^{-1}(\cdot)$ is the inverse of the Williams transformation.
Obtain an $m\times m$ matrix whose $i$th row is $h + (i-1)$ where $i = 1,\ldots,m$, and replace all the $0$ in the matrix with $m$. Denote $M$ the final matrix.   The following lemma is due to \cite{williams1949}, we also offer a proof in the appendix for completeness.

\begin{lemma}\label{lem:william_square}
  Let $m$ be an even positive integer; then the square matrix $M$ is a pair-balanced Latin square with entries from $\mathcal{Z}_m^+$.
\end{lemma}

In the following, we call $M$ the Williams Latin square. It can be used directly as the sequence design, but unlike the construction in \ref{constructnm}, the average absolute correlation $r_{ave}(M)$ can be large; see Example~\ref{ex:m=8}. To optimize $M$ according to the orthogonality criterion, we propose to use the level permutation method.

Let $O$ be an $m\times m$ sequence design with entries from $\mathcal{Z}_m^+$. By permuting the $m$ levels in $O$, we can obtain a new sequence design of the same size, denoted as $O'$. Obviously, $O$ and $O'$ have the same Hamming distance and pair-balance property, but maybe different $r_{ave}$ values. We use the threshold-accepting algorithm, outlined as Algorithm \ref{alg:searchO} in the appendix, to find the best sequence design.
Let $\tau=IJ$ be the number of iterations, where $I$ and $J$ are the numbers of outer and inner iterations, respectively.
Set the sequence of the thresholds as $T_1=\cdots=T_J$, $T_{J+1}=\cdots=T_{2J}=\gamma T_1$,
$T_{2J+1}=\cdots=T_{3J}=\gamma^2 T_1$, etc., where $\gamma=(T_{\tau}/T_1)^{\frac{1}{I-1}}$.
For more details on the parameter settings of the threshold accepting algorithm, refer to \cite{xiao2018construction}.

Now we construct the QS design $D=(X, O)$.
Suppose $m$ is an even number and there exists an integer $N$ such that $m=\phi(N)/2$.
We take the quantitative design $X$ to be $L$ in \eqref{eq:L} after replacing each $h_i$ in $L$ with $i$ for $i=1,\ldots,m$,
and the sequence design $O$ to be $M^*$, where $M^*$ is the level-permuted Williams Latin square output by Algorithm~\ref{alg:searchO} with $M$ as the initial design.
Then we obtain a QS design $D=(X,O)$ with $n=m$.

\begin{theorem}
\label{theo:XO_other_m}
	Let $D=(X,O)$ be constructed above. Then\\
	(i) $O$ is a Hamming equidistant and pair-balanced Latin square;\\
	(ii) if one of the following condition holds: (a) $N=p$ or $2p$, $m=\phi(N)/2=(p-1)/2$, $p$ an odd prime; (b) $N=4p$, $m=\phi(N)/2 = p-1$, $p$ an odd prime; (c) $N=2^t$, $t\ge 3$, $m=\phi(N)/2 = 2^{t-2}$, then
$X$ is an asymptotically maximin $L_1$-distance design which also has large $L_2$-distance; that is, as $m\rightarrow \infty$,
  $$d_1(X) \geq (1-\mathcal{O}(1/m)) d_{1,upper} \rightarrow d_{1,upper},$$
  $$d_2(X) \geq  \sqrt{{2}/{3}} (1-\mathcal{O}(1/m))  d_{2,upper} \rightarrow 0.817d_{2,upper},$$
  where the upper-bounds $d_{1,upper}$ and $d_{2,upper}$ are given in Lemma \ref{dupper}.
  In particular, $d_1(X)=d_{1,upper}$, that is, $X$ is $L_1$ -equidistant, when condition (a) holds.
\end{theorem}

\begin{example}\label{ex:m=8}
 Consider the construction of optimal QS design $D=(X,O)$ with $n=m=8$.
 First, we construct the quantitative design $X$. Let $N=17$ and $h=(1,\ldots,8)$.
 Then we can obtain an $8\times 8$ Latin square $L$ with entries from $\mathcal{Z}_8^+$ through \eqref{eq:L}. This design has $d_1(L)=24$ and $d_2(L) = \sqrt{90}$.
 Here $L$ achieves the bound $d_{1,upper}=24$ in Lemma~\ref{dupper}.

 Next, we construct the sequence design $O$. Let $h=(W^{-1}(0),W^{-1}(1),\ldots, \\
 W^{-1}(7)) = (0,7,1,6,2,5,3,4)$. Based on $h$
 we can construct an $8\times 8$ Williams Latin square $M$.
 It has $r_{ave}(M) = 0.333$, which is unsatisfied.
 Using $M$ as an initial design, we obtain a level-permuted design $M^*$ by Algorithm~\ref{alg:searchO}. This design has $r_{ave}(M^*) = 0.143$, which is smaller than $r_{ave}(M) = 0.333$.

 Finally, take $X=L$ and $O=M^*$, the the combined design $D = (X,O)$ is shown in Table~\ref{tab:ex:m=8}. The quantitative design $X$ is a maximin $L_1$-equidistant design and has large $L_2$-distance ($d_2(X)/d_{2,upper} = 0.968$). The sequence design $O$ has $d_H(O)=8$ (which equals to the $d_{H,upper}$ in Lemma~\ref{hupper}), $r_{ave}(O) = 0.143$, and all $t_{i,j}=1$ for $1\leq i\neq j \leq 8$.
\vspace{-0.1in}
 \begin{table}[htbp]
\begin{minipage}{\linewidth}
\caption{The optimal QS design with $n=m=8$.
\label{tab:ex:m=8}
}%
\begin{center}
\begin{tabular}{cccccccc c cccccccc}
			\multicolumn{8}{c}{$X$}  && \multicolumn{8}{c}{$O$}   \\
1 & 2 & 3 & 4 & 5 & 6 & 7 & 8 &  & 2 & 4 & 3 & 6 & 7 & 8 & 1 & 5 \\
2 & 4 & 6 & 8 & 7 & 5 & 3 & 1 &  & 3 & 2 & 7 & 4 & 1 & 6 & 5 & 8 \\
3 & 6 & 8 & 5 & 2 & 1 & 4 & 7 &  & 7 & 3 & 1 & 2 & 5 & 4 & 8 & 6 \\
4 & 8 & 5 & 1 & 3 & 7 & 6 & 2 &  & 1 & 7 & 5 & 3 & 8 & 2 & 6 & 4 \\
5 & 7 & 2 & 3 & 8 & 4 & 1 & 6 &  & 5 & 1 & 8 & 7 & 6 & 3 & 4 & 2 \\
6 & 5 & 1 & 7 & 4 & 2 & 8 & 3 &  & 8 & 5 & 6 & 1 & 4 & 7 & 2 & 3 \\
7 & 3 & 4 & 6 & 1 & 8 & 2 & 5 &  & 6 & 8 & 4 & 5 & 2 & 1 & 3 & 7 \\
8 & 1 & 7 & 2 & 6 & 3 & 5 & 4 &  & 4 & 6 & 2 & 8 & 3 & 5 & 7 & 1 \\
\end{tabular}
\end{center}
\end{minipage}
\end{table}	
\end{example}

For all even $m$ smaller than 100 satisfying that $m+1$ is not an odd prime, Table~\ref{tab:otherm} shows the distance ratios $d_1(X)/d_{1,upper}$ and $d_2(X)/d_{2,upper}$ and the $r_{ave}(O)$ values of the QS designs $(X,O)$ constructed by the method in this subsection. These designs fill the gap on the run size left by the method in Section~\ref{constructnm}.  We see that all of them have excellent performance under the maximin distance and orthogonality criteria.

\vspace{-0.1in}
\begin{table}[htbp]
\begin{minipage}{\linewidth}
\caption{Distance ratios $d_1(X)/d_{1,upper}$ and $d_2(X)/d_{2,upper}$ of quantitative designs and  $r_{ave}(O)$ values of sequence designs by our method for some small $m$ where $m+1$ is not an odd prime.
\label{tab:otherm}}%
\begin{center}
\begin{tabular}{cccc}\hline
$m$  & $d_1/d_{1,upper}$  & $d_2/d_{2,upper}$ & $r_{ave}$ \\
\hline
  8 & 1 & 0.968 &   0.143 \\
  14 & 1 & 0.958 &   0.077 \\
  20 & 1 & 0.954  &  0.053  \\
  24 & 0.930 & 0.913 &   0.043 \\
  26 &1 &0.951  &  0.040        \\
  32 &0.972 &0.929   & 0.032    \\
  44 &1 &0.948  &  0.023        \\
  48 &1 &0.948  &  0.021        \\
  50 &1 &0.947  &  0.020        \\
  54 &1& 0.947  &  0.019        \\
  56 &1 & 0.947 &   0.018 \\
  64 &0.986 & 0.936  &  0.016 \\
  68 &1& 0.946  &  0.015 \\
  74 & 1 &0.946 &   0.014\\
  80 & 0.977 &0.930  &  0.013\\
  84 & 0.978 & 0.931 &   0.012\\
  86 &1& 0.946   & 0.012\\
  90 &1 &0.945 &   0.011\\
  92 &0.980 &0.932  &  0.011\\
  98 &1 &0.945  &  0.010\\\hline
\end{tabular}
\end{center}
\end{minipage}
\end{table}

\section{Optimal QS design of run size \lowercase{$n=km$}}\label{sec4}

In this section, we present a construction method for optimal QS designs with run sizes $n=km$ for any $k \ge 2$, based on the constructions for $n=m$ in Section~\ref{sec3}.
The resulting designs are flexible to meet with various practical needs in one-shot experiments.

\begin{lemma}\label{lem:MCD}
   A QS design $D=(X,O)$ with run size $n=km$ for any $k \ge 2$ has the marginally coupled structure if and only if $O$ can be expressed as $O = \left(O_1^\T,O_2^\T,\ldots,O_k^\T\right)^\T$ up to some row permutations, where each $O_i$ is a $m\times m$ Latin square for $i=1,\ldots,k$.
\end{lemma}

By Lemma~\ref{lem:MCD}, for any marginally coupled design $D=(X,O)$, we can represent it as
   \begin{equation}
    \label{eq:XOpartition}
   D = (X,O) = \left(
   \begin{array}{cc}
     X_1 & O_1 \\
     \vdots & \vdots \\
     X_k & O_k
   \end{array}
   \right),
   \end{equation}
where each $O_i$ is an $m\times m$ Latin square. This motivates us to propose the following construction method.

We first present a three-step construction of optimal QS designs $D=(X,O)$ for $k= 2, \ldots, m+1$, based on the construction of $\tilde E_{b}$ in Section~\ref{constructnm}.
\setcounter{step}{0}
\begin{step}
Among all possible designs $\tilde E_{b}$ with $b \in \mathcal{Z}_p$, let $\tilde E_{b_1}$, \ldots, $\tilde E_{b_k}$ be the designs with the $k$ smallest $r_{ave}$ values. Obtain designs $O_1,\ldots, O_k$ by randomly permuting the rows of $\tilde E_{b_1}, \ldots, \tilde E_{b_k}$. Obtain the $km \times m$ sequence design by $O = (O_1^\T,\ldots,O_k^\T)^\T$.
\end{step}
\begin{step}
Let the matrix $F = (F_1^\T,\ldots,F_k^\T)^\T$ where each $F_i$, $i=1,\ldots,k$, is an $m \times m$ matrix obtained by randomly permuting the columns of $F_0$. Here $F_0=\tilde E_{b^*_2}$ with $b^*_2$ defined in Theorem~\ref{theo:XO_m} if $2m+1$ is not an odd prime and $F_0=E$ defined in Corollary~\ref{theo:XO_m_cor} if $2m+1$ is an odd prime.
\end{step}
\begin{step}
For $j=1,\ldots,m$, let $f_j$ be the $j$th column of $F$ and let $\ell_j$ be a random permutation of $(0,1,\ldots,k-1)^\T$. Let the $j$th column of the quantitative design $X$ be $m \ell_j \otimes 1_m +  f_j $, where $1_m$ is a vector of $m$ ones and $\otimes$ is the Kronecker product. Obtain the $km \times 2m$ optimal QS design $D=(X,O)$.
\end{step}

\begin{theorem}\label{theo:XOkm}
The QS design $D=(X,O)$ with run size $n=km$, $2\leq k \leq m+1$, constructed by this three-step method has the following properties: \\
(i) $O$ is a pair-balanced design, where all $t_{i,j}=k$ for $1\leq i \neq j \leq m$; \\
(ii) $O$ is an asymptotically maximin Hamming distance design:
$$d_H(O) \geq m-3 = (1-\mathcal{O}(1/m)) d_{H,upper} \rightarrow d_{H,upper} \text{ \ as \ } m\rightarrow \infty,$$
where the upper-bound $d_{H,upper}$ is defined in Lemma \ref{hupper}; \\
(iii) $O$ is an asymptotically orthogonal design with
$$r_{ave}(O) < \mathcal{O}(1/m) \rightarrow 0 \text{ \ as \ } m\rightarrow \infty;$$
(iv) $X$ has desirable space-filling properties with
$$d_1(X)\geq (1-\mathcal{O}(1/m))  \lfloor (m+1)m /3\rfloor,$$
$$d_2(X) \geq \sqrt{{2}/{3}} (1-\mathcal{O}(1/m))  \sqrt{ \lfloor  m^2(m+1)/6 \rfloor }; $$
(v) $D=(X,O)$ is a marginally coupled design.
\end{theorem}

\begin{example}\label{ex0}
Consider the construction of a QS design with $n=12$ and $m=4$, where $p=m+1=5$ and $k=n/m=3$. The 3-steps are illustrated as below.
	
	Step~1. From Table~\ref{tab:b_and_b1}, all the following three designs have the same smallest $r_{ave}$ value of $0.333$ among all $\tilde E_b$ with $b \in \mathcal{Z}_5$ when $p=5$. Let $O_1= \tilde E_{1}$, $O_2= \tilde E_{3}$, $O_3= \tilde E_{4}$ and $O = (O_1^\T,O_2^\T,O_3^\T)^\T$, where
$$
	\tilde E_{1} =   \begin{pmatrix}
		4 & 3 & 2 & 1 \\
		3 & 1 & 4 & 2 \\
		2 & 4 & 1 & 3 \\
		1 & 2 & 3 & 4 \\
		\end{pmatrix},
	\tilde E_{3} = \begin{pmatrix}
		2 & 1 & 3 & 4 \\
		1 & 4 & 2 & 3 \\
		3 & 2 & 4 & 1 \\
		4 & 3 & 1 & 2 \\
		\end{pmatrix},
	\tilde E_{4} =  \begin{pmatrix}
		1 & 2 & 4 & 3 \\
		2 & 3 & 1 & 4 \\
		4 & 1 & 3 & 2 \\
		3 & 4 & 2 & 1 \\
		\end{pmatrix}.
	$$

\medskip\medskip
	
	Step~2. By Table~\ref{tab:b_and_b1}, WLOG, we select $b^*_2=3$.	Let $F = (F_1^\T,F_2^\T,F_3^\T)^\T$, where each $F_i$ is obtained by randomly permuting the columns of $F_0=\tilde E_{3}$:	
	$$F_1 =  \begin{pmatrix}
		3 & 4 & 1 & 2 \\
		2 & 3 & 4 & 1 \\
		4 & 1 & 2 & 3 \\
		1 & 2 & 3 & 4 \\
		\end{pmatrix},
	F_2 =  \begin{pmatrix}
		4 & 3 & 2 & 1 \\
		3 & 2 & 1 & 4 \\
		1 & 4 & 3 & 2 \\
		2 & 1 & 4 & 3 \\
		\end{pmatrix},
	F_3 =  \begin{pmatrix}
		4 & 2 & 1 & 3 \\
		3 & 1 & 4 & 2 \\
		1 & 3 & 2 & 4 \\
		2 & 4 & 3 & 1 \\
		\end{pmatrix}. $$
\medskip\medskip
	
  Step~3. For $j=1,\ldots,4$, let
		$\ell_1 = (0,2,1)^\T$, $\ell_2 = (2,0,1)^\T$, $\ell_3 =(1,0,2)^\T$ and $\ell_4 =(0,2,1)^\T$
		be four random permutations of $(0,1,2)^\T$. Obtain the $j$th column of $X$ as $4 \ell_j \otimes 1_4 +  f_j $, where $f_j$ is the $j$th column of $F$. Finally, we obtain the design $D=(X,O)$ in Table~\ref{tab:exam2}.

\begin{table}[htbp]
\begin{minipage}{\linewidth}
\caption{A QS design with $n=12$ and $m=4$.
\label{tab:exam2}}%
\begin{center}
\begin{tabular}{cccc c cccc}
\multicolumn{4}{c}{$X$}  && \multicolumn{4}{c}{$O$}   \\
				3 & 12 & 5 & 2 && 4 & 3 & 2 & 1 \\
				{\bf 2} & 11 & 8 & 1 && 3 & {\bf 1} & 4 & 2 \\
				4 & 9  &  6 & 3 && 2 & 4 & 1 & 3 \\
				1 & 10 & 7 & 4 && 1 & 2 & 3 & 4 \\ \hdashline[2pt/2pt]
				{\bf 12} & 3 & 2 & 9 && 2 & {\bf 1} & 3 & 4 \\
				11 & 2 & 1 & 12 && 1 & 4 & 2 & 3 \\
				9 & 4 & 3 & 10 && 3 & 2 & 4 & 1 \\
				10 & 1 & 4 & 11 && 4 & 3 & 1 & 2 \\ \hdashline[2pt/2pt]
				8 & 6 & 9 & 7 && 1 & 2 & 4 & 3 \\
				7 & 5 & 12 & 6 && 2 & 3 & 1 & 4 \\
				{\bf 5} & 7 & 10 & 8 && 4 & {\bf 1} & 3 & 2 \\
				6 & 8 & 11 & 5 && 3 & 4 & 2 & 1 \\
\end{tabular}
\end{center}
\end{minipage}
\end{table}

This design $D$ in Table~\ref{tab:exam2} has a marginally coupled structure. As an illustration, when component 1 (i.e. $c_1$) is in the second order (i.e. the second column of $O$), its quantitative levels (i.e. the first column of $X$) are $2,12$ and $5$, which are equally spaced if the $12$ levels of $X$ are collapsed to $3$ levels. This also holds for the quantitative levels corresponding to the other three components ($c_2$, $c_3$ and $c_4$).
\end{example}

For certain design sizes, in addition to the properties shown in Theorem \ref{theo:XOkm}, the constructed designs can be further improved. Here, we illustrate the case of $k=2$.

\begin{lemma}\label{lemma:MCD2}
For any marginally coupled design $D=(X,O)$ with run size $n=2m$, if $d_H(O) \geq m-2$ and the number of run pairs having the Hamming distance of $m-2$ is smaller than $m^2/2$, the $L_1$- and $L_2$-distances of the quantitative design $X$ have the upper-bounds:
  	$$
 	d_1(X) \leq d_{1,upper} = \lfloor (m+1)m /3\rfloor,
  $$
  $$
    d_2(X) \leq d_{2,upper}  = \sqrt{ \lfloor  m^2(m+1)/6 \rfloor }.
  $$
\end{lemma}

As we seek space-filling sequence designs $O$ under the Hamming distance, the conditions of Lemma~\ref{lemma:MCD2} always hold. As shown in the following Proposition~\ref{prop:MCD2}, when $n=2m$, the above three-step construction can be simplified and the resulting designs will have improved space-filling properties.

\begin{proposition}\label{prop:MCD2}
When $n=2m$, let the sequence design $O=(O_1^\T,O_2^\T)^\T$ where $O_1 = \tilde E_{b_1^*}$ with $b^*_1$ determined by \eqref{b1_star} and $O_2 = \tilde E_{p-b_1^*}$. Follow Steps 2 and 3 to construct the quantitative design $X$. The resulting marginally coupled design $D=(X,O)$ has the following properties: \\
(i) $O$ has improved minimum pairwise Hamming distance: $d_H(O) = m-2$;
\\
(ii) $X$ is an asymptotically maximin $L_1$-distance design which also has large $L_2$-distance:
  $$d_1(X) \geq (1-\mathcal{O}(1/m)) d_{1,upper} \rightarrow d_{1,upper},$$
  $$d_2(X) \geq  \sqrt{{2}/{3}} (1-\mathcal{O}(1/m))  d_{2,upper} \rightarrow 0.817 d_{2,upper},$$
  where $m\rightarrow \infty$ and the upper-bounds $d_{1,upper}$ and $d_{2,upper}$ are defined in Lemma~\ref{lemma:MCD2}.
\end{proposition}

\begin{example}\label{ex1}
	In Table~\ref{tab:exam}, we show the QS design $D = (X,O)$ with $n=12$ runs and $m=6$ components constructed according to Proposition~\ref{prop:MCD2} where $b^*_1=1$ in Table~\ref{tab:b_and_b1} is used. As $2m+1=13$ is an odd prime, in Step~2, $F_0=E$ with $E$ the $L_1$-equidistant LHD defined in Corollary~\ref{theo:XO_m_cor}.
The design $O$ is pair-balanced with all $t_{i,j} = 2$ for $1\leq i \neq j \leq 6$, and has $d_H(O)=4$. The design $X$ is an LHD with $d_1(X) = 14$ and $d_2(X) = \sqrt{40}$. 
Here, $d_1(X)$ achieve the upper-bound in Lemma~\ref{lemma:MCD2}, and thus $X$ is a maximin $L_1$-distance design.
The design $D = (X,O)$ is marginally coupled, as shown by the dashed line in Table~\ref{tab:exam}.
\vspace{-0.1in}
\begin{table}[htbp]
\begin{minipage}{\linewidth}
\caption{A QS design with $n=12$ and $m=6$.
\label{tab:exam}}%
\begin{center}
\begin{tabular}{cccccc c cccccc}
\multicolumn{6}{c}{$X$}  && \multicolumn{6}{c}{$O$}   \\
				3 & 6 & 2 & 10 & 5 & 1 &&  4 & 6 & 5 & 3 & 2 & 1 \\
				6 & 1 & 4 & 11 & 3 & 2 && 6 & 3 & 1 & 4 & 5 & 2 \\
				4 & 5 & 6 & 7 & 2 & 3 && 5 & 1 & 6 & 2 & 4 & 3 \\
				1 & 2 & 5 & 9 & 6 & 4 && 3 & 4 & 2 & 6 & 1 & 5 \\
				2 & 4 & 3 & 12 & 1 & 5 && 2 & 5 & 4 & 1 & 3 & 6 \\
				5 & 3 & 1 & 8 & 4 & 6 && 1 & 2 & 3 & 5 & 6 & 4 \\ \hdashline[2pt/2pt] 
				11 & 10 & 8 & 6 & 7 & 9 && 1 & 2 & 4 & 6 & 5 & 3 \\
				9 & 11 & 10 & 1 & 8 & 12 && 2 & 6 & 3 & 1 & 4 & 5 \\
				8 & 7 & 12 & 5 & 9 & 10 && 4 & 3 & 2 & 5 & 1 & 6 \\
				12 & 9 & 11 & 2 & 10 & 7 && 6 & 1 & 5 & 2 & 3 & 4 \\
				7 & 12 & 9 & 4 & 11 & 8 && 5 & 4 & 1 & 3 & 6 & 2 \\
				10 & 8 & 7 & 3 & 12 & 11 && 3 & 5 & 6 & 4 & 2 & 1 \\
\end{tabular}
\end{center}
\end{minipage}
\end{table}
\end{example}
\vspace{-0.3in}

For general values of $m$ and $k\ge 2$,  we can use a modification of Steps~1--3 to obtain QS designs with a marginally coupled structure.
Suppose that we have two $m\times m$ Latin squares $X_0$ and $O_0$, where $X_0$ is a quantitative design and $O_0$ is a pair-balanced sequence design. We modify Steps~1--2  as follows.

{\bf Step~1'} Let $O_{init} = (O_1^T, \ldots, O_k^T)$ be a $km\times m$ sequence design, where $O_i, i=1,\ldots,k$ are obtained by
randomly permuting the levels in $O_0$. Optimize $O_{init}$ using Algorithm~\ref{alg:searchOk} in the appendix to obtain the sequence design $O$.

{\bf Step~2'} Let the matrix $F = (F_1^\T,\ldots,F_k^\T)^\T$ where each $F_i$, $i=1,\ldots,k$, is an $m \times m$ matrix obtained by randomly permuting the columns of $X_0$.

Proceeding with Step~3, we can obtain the desired $km \times 2m$ QS design $D=(X,O)$
which by Lemma~\ref{lem:MCD} is a marginally coupled design.
The sequence design $O$ is a pair-balanced design with all $t_{i,j}=k$ for $1\leq i \neq j \leq m$.

In Step~1', we need to optimize the initial sequence design $O_{init} = (O_1^T, \ldots, O_k^T)$ under the maximin Hamming distance and orthogonality criteria simultaneously.
We use the level permutation method to minimize $w r_{ave}(O) + (1-w) (1-d_H(O)/d_{H,upper})$, that is, a weighted criterion, where $w\in [0,1]$ is a user-specified weight.
Since $O_{0}$ is pair-balanced, the design $O_{init}$ will be pair-balanced after level permutations.
By Lemma~\ref{hupper}, $d_{H,upper}=m-1$ for $k\ge 2$. Because both $r_{ave}(O)$ and $d_H(O)/d_{H,upper}$ are $0$-$1$ scaled,
we simply take $w=0.5$. A threshold accepting algorithm is described in Algorithm~\ref{alg:searchOk} in the appendix.

Specifically, suppose $m$ is an even number and there exists an integer $N$ such that $m=\phi(N)/2$. Then two Latin squares $X_0$ and $O_0$ in Step~1' and~2' can be taken as the designs constructed in Section~\ref{sec3.2}, that is, $X_0 = L$ defined in \eqref{eq:L} after replacing each $h_i$ in $L$ with $i$ for $i=1,\ldots,m$, and $O_0=M$ where $M$ is the Williams Latin square.

\begin{proposition}\label{prop:other_n=km}
Suppose $m$ is an even number and there exists an integer $N$ such that $m=\phi(N)/2$.
The QS design $D=(X,O)$ with run size $n=km$, $k\ge 2$, constructed by Steps~1', 2' and~3 using  $X_0 = L$ and $O_0=M$ has the following properties: \\
(i) $O$ is a pair-balanced design, where all $t_{i,j}=k$ for $1\leq i \neq j \leq m$; \\
(ii) if one of the following condition holds: (a) $N=p$ or $2p$, $m=\phi(N)/2=(p-1)/2$, $p$ an odd prime; (b) $N=4p$, $m=\phi(N)/2 = p-1$, $p$ an odd prime; (c) $N=2^t$, $t\ge 3$, $m=\phi(N)/2 = 2^{t-2}$, then $X$ has desirable space-filling properties with
$d_1(X)\geq (1-\mathcal{O}(1/m))  \lfloor (m+1)m /3\rfloor$ and
$d_2(X) \geq \sqrt{{2}/{3}} (1-\mathcal{O}(1/m))  \sqrt{ \lfloor  m^2(m+1)/6 \rfloor }$. \\
(iii) $D=(X,O)$ is a marginally coupled design.
\end{proposition}

\begin{example}\label{ex:16times8}
 Using Steps~1', 2' and~3, we obtain a marginally coupled QS design $D = (X,O)$ with $n=16$ runs and $m=8$ components, as shown in Table~\ref{tab:exam:16times8}. This design has $d_1(X) = 24$, $d_2(X) = \sqrt{90}$,  $d_H(O)=6$ and $r_{ave}(O) =0.143$. It can be checked that the condition in Lemma~\ref{lemma:MCD2} holds,  thus $d_1(X)$ attains the bound $d_{1,upper}=\lfloor (m+1)m /3\rfloor$.

 \begin{table*}[htbp]
\caption{A QS design with $n=16$ and $m=8$.
\label{tab:exam:16times8}}
\tabcolsep=0pt
\begin{tabular*}{\textwidth}{@{\extracolsep{\fill}}cccccccc c cccccccc @{\extracolsep{\fill}}}
\multicolumn{8}{c}{$X$}  && \multicolumn{8}{c}{$O$}   \\
7 & 8 & 3 & 6 & {10} & 4 & 5 & 1 &  & 3 & 5 & 7 & 1 & 8 & 6 & 2 & 4 \\
3 & 1 & 6 & 5 & {12} & 8 & 7 & 2 &  & 7 & 3 & 8 & 5 & 2 & 1 & 4 & 6 \\
4 & 7 & 8 & 1 & {14} & 5 & 2 & 3 &  & 8 & 7 & 2 & 3 & 4 & 5 & 6 & 1 \\
6 & 2 & 5 & 7 & {16} & 1 & 3 & 4 &  & 2 & 8 & 4 & 7 & 6 & 3 & 1 & 5 \\
1 & 6 & 2 & 4 & {15} & 3 & 8 & 5 &  & 4 & 2 & 6 & 8 & 1 & 7 & 5 & 3 \\
8 & 3 & 1 & 2 & {13} & 7 & 4 & 6 &  & 6 & 4 & 1 & 2 & 5 & 8 & 3 & 7 \\
2 & 5 & 4 & 8 & {11} & 6 & 1 & 7 &  & 1 & 6 & 5 & 4 & 3 & 2 & 7 & 8 \\
5 & 4 & 7 & 3 & {9} & 2 & 6 & 8 &  & 5 & 1 & 3 & 6 & 7 & 4 & 8 & 2 \\ \hdashline[2pt/2pt] 
{15} & {16} & {11} & {14} & 2 & {12} & {13} & {9} &  & 6 & 5 & 2 & 4 & 8 & 1 & 7 & 3 \\
{11} & {9} & {14} & {13} & 4 & {16} & {15} & {10}  &  & 2 & 6 & 8 & 5 & 7 & 4 & 3 & 1 \\
{12} & {15} & {16} & {9} & 6 & {13} & {10} & {11}  &  & 8 & 2 & 7 & 6 & 3 & 5 & 1 & 4 \\
{14} & {10} & {13} & {15} & 8 & {9} & {11} & {12}  &  & 7 & 8 & 3 & 2 & 1 & 6 & 4 & 5 \\
{9} & {14} & {10} & {12} & 7 & {11} & {16} & {13}  &  & 3 & 7 & 1 & 8 & 4 & 2 & 5 & 6 \\
{16} & {11} & {9} & {10} & 5 & {15} & {12} & {14}  &  & 1 & 3 & 4 & 7 & 5 & 8 & 6 & 2 \\
{10} & {13} & {12} & {16} & 3 & {14} & {9} & {15}  &  & 4 & 1 & 5 & 3 & 6 & 7 & 2 & 8 \\
{13} & {12} & {15} & {11} & 1 & {10} & {14} & {16}  &  & 5 & 4 & 6 & 1 & 2 & 3 & 8 & 7 \\
\end{tabular*}
\end{table*}
 \end{example}

\section{Numerical Study}\label{sec5}

The Travel Salesman Problem (TSP) is a well-known optimization challenge. \cite{xiao2021} considered a TSP variant with quantitative-sequential (QS) inputs. Suppose a traveling salesman departs from city $0$ at time $0$ and must visit $m$ cities (denoted as cities $1, \ldots, m$) to complete a sales task. The travel time from city $i$ to city $j$ ($i \neq j$) is given by $s_{i,j}$ days, where $s_{i,j}$ may not necessarily equal $s_{j,i}$. Each city $i$ has a predefined deadline $d_i$ ($i=1,\ldots,m$) by which the sales task must be completed. If the salesman fails to complete the task on time, a penalty of $f$ currency units per day is incurred. Upon completing the sales tasks in all cities, the salesman earns a fixed income of $ma$ currency units (with $a$ units per city) and a commission, where each city contributes $e$ units per day of sales. Additionally, the salesman incurs a daily travel expense of $b$ currency units. The objective is to determine an optimal strategy, that is, the sequence of city visits and the duration of stay in each city to maximize total profit.

Any strategy in this problem can be represented as a QS input, denoted by $w = (x, o)$, where the quantitative component $x = (x_1, \ldots, x_m)$ represents the duration of stay in each city $i$ ($i=1,\ldots,m$), and the sequential component $o = (o_1, \ldots, o_m)$ represents the order of visitation of the cities (ie, a sequence of city indices). Let $o_0 = 0$ denote the starting city, then the completion time of the sales task in city $o_i$ is given by  
$$
C(w,o_i) = \sum_{l=1}^{i} (s_{o_{l-1},o_l} + x_{o_l}).
$$  
The delay time for city $o_i$ is  
$$
T(w, o_i) = \max(0, C(w,o_i) - d_{o_i}).
$$  
Thus, the total profit under strategy $w$ is  
$$
F(w) = m a + e \sum_{i=1}^{m} x_i - b C(w, o_m) - f \sum_{j=1}^{m} T(w, o_j),
$$  
where the detailed derivation can be found in \cite{xiao2021}.

This section considers a travel salesman problem with $m=6$ cities.  
Before carrying out the experiment, we assume that $F(w)$ is unknown. \cite{xiao2021} proposed a Bayesian optimization method, termed QS-learning, to address this optimization problem. Bayesian optimization begins with an initial design and then sequentially selects the next experiment point based on a statistical model (such as the MaGP model) and an acquisition function. This iterative process aims to identify the input $w^*$ that maximizes $F(w)$. The simulation results in \cite{xiao2021} demonstrate that their proposed initial design, combined with the MaGP model, outperforms random designs and other modeling approaches such as PWO and CP in efficiently approximating the optimal solution.

Following the parameter settings in \cite{xiao2021}, we set $m=6$, $a=20$, $e=10$, $b=2$, and $f=15$. The sales deadlines for each city are given by
$
(d_1, \ldots, d_6) = (26, 10, 23, 25, 12, 10).
$
The allowable stay duration in each city is constrained by $x_i \in [1,4]$, and the travel time matrix between cities is as follows:  
\begin{align*}\footnotesize
\big(s_{i,j}\big)_{i=0,\ldots,m, j=1,\ldots,m} =
\begin{pmatrix}
 0.6 & 2.2 & 1.8 & 2.6 & 1.8 & 1.7 \\
 0 & 0.8 & 1.5 & 1.4 & 2.8 & 1.1 \\
 0.7 & 0 & 1.2 & 2.4 & 2.3 & 1.4 \\
 1.5 & 1.2 & 0 & 1.8 & 1.3 & 1.5 \\
 1.2 & 2.4 & 1.7 & 0 & 1.7 & 2.1  \\
 2.7 & 2.4 & 1.3 & 1.7 & 0 & 0.9  \\
 1.1 & 1.3 & 1.4 & 2.3 & 0.9 & 0  \\
\end{pmatrix}.	
\end{align*}

Under the setting of $n = m = 6$, we compare the performance of the following two initial designs in QS-learning: (i) the initial design constructed in Example \ref{ex:m=6} of this paper (see Table \ref{tab:exam1}), denoted by $D_1 = (X_1, O_1)$; and (ii) the initial design constructed using the method proposed in Theorem 2 of \cite{xiao2021}, denoted as $D_2 = (X_2, O_2)$, where both $X_2$ and $O_2$ are the first six rows of $D_0$ in Example \ref{ex:m=6}. Although both initial designs exhibit favorable properties, $D_1$ demonstrates superior performance in terms of distance and orthogonality. Specifically, the quantitative design $X_1$ in $D_1$ outperforms $X_2$ in $D_2$ under the max-min $L_1$ and $L_2$ distance criteria, i.e., $ d_1 (X_1) = 14 > d_1(X_2) = 12, \quad d_2(X_1) = \sqrt{40} > d_2(X_2) = \sqrt{28}.$ 
Both sequence designs $O_1$ and $O_2$ satisfy the pairwise balance property, but the average correlation of $O_1$ is lower than that of $O_2$, i.e.,  
$
r_{ave}(O_1) = 0.2 < r_{ave}(O_2) = 0.36.
$ 
Based on these two initial designs, we applied the QS-learning algorithm to sequentially collect 24 additional experiment points. The modeling in QS-learning adopts the MaGP model, and Bayesian optimization uses the Expected Improvement (EI) acquisition function. For further details, see Sections 3 and 4 of \cite{xiao2021}. The left panel of Figure \ref{fig2} shows the cumulative maximum response values under the two initial designs. We observe that using $D_2$ as the initial design achieves a high profit value of
$
F(w_2^*) = 207.37
$
at the 17th sequential point $w_2^* = (1.08, 1, 1.76, 3.34, 4, 2.75, 1, 6, 2, 5, 3, 4)$, employing $D_1$ exceeds this value earlier (at the 9th sequential point). Moreover, at the 24th sequential point, it identifies the strategy $ w_1^* = (1.35, 2.09, 2.23, 2.75, 4, 2.81, 6, 2, 5, 3, 1, 4), $ that produces a superior profit value of $ F (w_1 *) = 222.84$. This result demonstrates that adopting the initial design proposed in this paper further enhances the convergence speed and final solution quality of the QS active learning method proposed by \cite{xiao2021}.  

For comparison, we also generated a random design with $n = 300$ points, where the sequence design was randomly sampled without replacement from the $6! = 720$ possible city orderings, and the quantitative design was drawn from a uniform distribution. The histogram of the profit values obtained from this design is shown in the right panel of Figure \ref{fig2}. The maximum profit achieved by this random design was $ 204.93$,
which is lower than the results obtained using QS active learning with only 30 samples based on either $D_1$ or $D_2$.

\begin{figure*}[htbp]
\centering
\includegraphics[angle=0, scale=0.5]{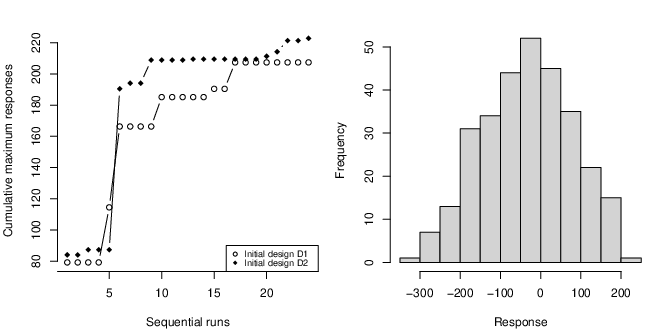}\par
\caption{Left panel: Cumulative maximum response values in QS-learning using the initial designs $D_1$ (triangles) and $D_2$ (circles) for $n = m = 6$.  
Right panel: Histogram of the response values obtained using a random design with $n = 300$.}\label{fig2}
\end{figure*}

\section{Discussion}\label{sec6}
In this paper, we construct a new class of optimal quantitative-sequence (QS) designs with run sizes $n=km$, $k\ge 1$, for $m$ components; see Table~\ref{tab_app} in the appendix for a catalog of some small-size designs obtained.
We prove theoretical results on their marginally coupled, pair-balanced, space-filling and asymptotically orthogonal properties.
We also present some generalizations of the proposed construction to gain more flexibility in design sizes. In a simulation study on the traveling salesman problem, we demonstrate that the constructed design improves active learning performance compared to the design proposed by \cite{xiao2021}.

For future research, we will explore whether similar theoretical results may hold for any other $m$. In addition, for a general $m$, one can first construct a larger design with $m' = p-1 > m$, and then delete some runs and columns to obtain the required sizes.
This method is a generalization of the leave-one-out approach and was adopted by \cite{wang2018optimal}. Their theoretical and simulation results indicate that the resulting designs generally exhibit desirable distance and orthogonality properties. Although the pairwise balance property may not be fully preserved, when $m'$ is large and the number of deleted rows and columns is small, the $t_{i,j}$ values of the obtained design will be approximately equal.
The $X$ part of our design is an LHD, which is desirable for studying quantitative factors in computer experiments. When there are restrictions on the number of levels for quantitative factors, for example, in some physical experiments, we can collapse the $n$ levels of $X$ to obtain a balanced fractional factorial design $X'$. A similar marginally coupled structure will still hold for $(X',O)$, that is, for each level of any factor in $O$ the corresponding design points in $X'$ also form a small balanced design. We will further study other space-filling properties for such a $(X', O)$.

Another interesting topic is investigating how to further improve the quantitative design $X$ for $k \geq 3$ as we have done for the cases of $k=1$ and $k=2$ in this paper. Our proposed designs are very attractive for high-dimensional cases, since they are asymptotically optimal under orthogonal and distance criteria. Yet, for small cases, we can still improve them via a structured level permutation technique.
For example, Theorem \ref{theo:XOkm} and Proposition \ref{prop:other_n=km} only guarantee the asymptotic lower bounds of the $L_1$ and $L_2$ distances of $X$. We will further investigate how to improve the space filling property of $X$ while preserving the marginal coupling structure of $D$. One possible approach is to improve the distance properties of $X$ using random optimization algorithms and level permutation techniques in Steps 4.2–4.3.

The proposed designs aim to provide robust performances that suit various types of model given their desirable structural properties. For example, the proposed designs with $n=m$ can be used as efficient initial designs in active learning under the Gaussian process models by \cite{xiao2021}. For one-shot experiments, the pairwise ordering (PWO) model \citep{van1995design, voelkel2019design} which are used to model order-of-addition factors can be generalized to also include quantitative factors, where the proposed designs with run sizes $n=m(m+1)$ suffice. Given the limited literature on this new type of experiments, we will also work on developing new modeling techniques and studying designs' efficiency under the efficient models.

\section{Acknowledgments}
The authors thank anonymous reviewers for their valuable suggestions.

\section{Supplemental Online Material}

\textbf{Appendix:} This file contains additional technical details and the proofs of all theorems. (.pdf file)

\noindent
\textbf{Codes: }This file contains R codes and relevant files to perform the methods in the article. (zipped file)

\bibliographystyle{asa}
\bibliography{reference}

\newpage

\setcounter{page}{1}

  \begin{center}
    {\LARGE\bf Online Supplementary Materials for ``Design of Experiment with Quantitative-Sequence Factors''}
\end{center}

\begin{appendices}
\section{Design catalogue and algorithms}
In Table~S1 we provide a catalogue of the obtained  designs with small run sizes and the related construction methods (sources).
Algorithm~\ref{alg:searchO} is the threshold accepting algorithm used in Section~\ref{sec3.2} for searching a best sequence design.
Algorithm~\ref{alg:searchOk} is the threshold accepting algorithm used in Section~\ref{sec4} for searching a best sequence design. In each iteration of Algorithm~\ref{alg:searchOk}, we randomly choose one subarray $O_j$, $1 \le j \le k$ and permute its levels (see line~5).

 \begin{table*}[htbp]
\caption{A catalogue of the obtained  designs with small run sizes ($m \le 20, n\le 50$).}
\label{tab_app}
\tabcolsep=0pt
\begin{tabular*}{\textwidth}{@{\extracolsep{\fill}}cccc @{\extracolsep{\fill}}}
\hline
 Number of factors & Number of runs & \multicolumn{2}{c}{Source}  \\
  $m$ & $n$  & $n=m$ & $n>m$ \\ \hline
 4 & 4, 8, 12, 20  & Theorem~\ref{theo:XO_m} & Theorem~\ref{theo:XOkm}\\
 6 & 6, 12, 18, 24, 30, 36, 42  & Corollary~\ref{theo:XO_m_cor}  & Theorem~\ref{theo:XOkm}\\
 8 & 8, 16, 24, 32, 40, 48 & Theorem~\ref{theo:XO_other_m} & Proposition~\ref{prop:other_n=km} \\
 10 & 10, 20, 30, 40, 50  & Theorem~\ref{theo:XO_m} & Theorem~\ref{theo:XOkm}\\
 12 & 12, 24, 36, 48 & Theorem~\ref{theo:XO_m} & Theorem~\ref{theo:XOkm}\\
 16 & 16, 32, 48 & Theorem~\ref{theo:XO_m} & Theorem~\ref{theo:XOkm}\\
 18 & 18, 36 & Corollary~\ref{theo:XO_m_cor}  & Theorem~\ref{theo:XOkm}\\
 20 & 20, 40 & Theorem~\ref{theo:XO_other_m} & Proposition~\ref{prop:other_n=km} \\ \hline
\end{tabular*}
\end{table*}

\begin{algorithm}[htbp]
\caption{Pseudo code for optimizing $O$ with $n=m$.}\label{alg:searchO}
\begin{algorithmic}[1]
\State Initialize $\tau$ and the sequence of thresholds $T_1,\ldots,T_{\tau}$.
\State Input an initial sequence-design $O$ and let $O^{*}=O$. Compute $r_{ave}(O)$.
\For{$i=1$ to $\tau$}
    \State Randomly choose $j, k \in \mathcal{Z}_m^+$, $j\ne k$.
    \State Obtain $O'$ by exchanging the levels $j$ and $k$ in $O$.
    \If {$r_{ave}(O') < (1+T_i) r_{ave}(O)$,}
        \State Update $O=O'$ and $r_{ave}(O)=r_{ave}(O')$.
        \State \textbf{ if } $r_{ave}(O)<r_{Ave}(O^{*}),$ \textbf{ then } update $O^*=O$, \textbf{ end if}
    \EndIf
\EndFor
\end{algorithmic}
\end{algorithm}

\begin{algorithm}[htbp]
\caption{Pseudo code for optimizing $O$ with $n=km$, $k\ge 2$.}\label{alg:searchOk}
\begin{algorithmic}[1]
\State Initialize $\tau$ and the sequence of thresholds $T_1,\ldots,T_{\tau}$.
\State Input the initial sequence-design $O_{init} = (O_1^T, \ldots, O_k^T)$, where each $O_i$ is a $m\times m$ Latin square, $i=1,\ldots,k$.
\State Let $O^{*}=O$. Compute $\psi(O) = 0.5r_{ave}(O)+0.5(1-d_H(O)/(m-1))$.
\For{$i=1$ to $\tau$}
    \State Randomly choose $j \in \{1,\ldots, k\}$.
    \State Randomly choose $l_1, l_2 \in \mathcal{Z}_m^+$, $l_1\ne l_2$.
    \State Obtain $O'$ by exchanging the levels $l_1$ and $l_2$ in $O_j$ of $O$.
    \If {$\psi(O')< (1+T_i) \psi(O)$,}
        \State Update $O=O'$ and $\psi(O)=\psi(O')$.
        \State \textbf{ if } $\psi(O)<\psi(O^{*}),$ \textbf{ then } update $O^*=O$, \textbf{ end if}
    \EndIf
\EndFor
\end{algorithmic}
\end{algorithm}

\section{Proofs}

\begin{proof}[Proof of Theorem~\ref{lem:balance}]
	Parts (i) and (iii) follow directly by the constructions of good lattice point sets and the leave-one-out method \citep{zhou2015space}. Now we show part (ii). When the leave-one-out design $\tilde D_0$ is used as the sequence design, for any $1\leq i \leq m$, the component ``$i$'' appears exactly once in each of the $m$ rows of $\tilde D_0$ by the construction of good lattice point set. Because $p$ is a prime number, for each $k \in \mathcal{Z}_m^+$, there is a unique $x \in \mathcal{Z}_m^+$ such that $k x = i \text{ (mod } p\text{)}$. This implies that for each $1\leq k \leq m$ and $k \neq p-i$, the sub-sequence ``$i, j$'' appears in the $k$th row of $\tilde D_0$ where $j = i+k \text{ (mod } p\text{)}$. As $j = i+k \text{ (mod } p \text{)} $ enumerates elements in $\mathcal{Z}_m^+ \setminus \{i\}$ for $k = 1\ldots m$ and $k \neq p-i$, the subsequence ``$i, j$'' appears exactly once in all rows of $\tilde D_0$ for all $j=1,\ldots, m$ and $j\neq i$, that is, $\tilde D_0$ is pair-balanced.
For any $b\in \mathcal{Z}_p =  \left\{0,\ldots,p-1\right\}$, the leave-one-out design $\tilde D_b$ can be transformed by a level permutation from $\tilde D_0$, and the leave-one-out design $\tilde E_b$ can be transformed by a level permutation from $\tilde D_b$. Therefore, they are all pair-balanced when used as sequence designs, and the conclusion follows.
\end{proof}

\begin{proof}[{Proof of Theorem~\ref{theo:XO_m}}]
	We first show part (i). The Hamming equidistant and pair-balance properties are due to Theorem~\ref{lem:balance}. The sequence design  $O$ is chosen as a special leave-one-out design $\tilde E_{b_1^*}$ with the least average absolute correlation value, thus by the following Theorem~\ref{lem:r} from \cite{wang2018optimal},   $r_{ave}(\tilde E_{b_1^*}) < 5(m+2)/(m-1)^2 = \mathcal{O}(1/m) \rightarrow 0 \text{  as  } m\rightarrow \infty.$

\begin{lemma}\label{lem:r}
  Let $p$ be an odd prime, $D$ be an $p\times (p-1)$ good lattice point design, $D_b=D+b \text{ (mod } p)$, $E_b = W(D_b)$, and $\tilde{E}_b$ be the leave-one-out design obtained from $E_b$ for $b=0,1,\ldots,p-1$. Then $r_{ave}(\tilde{E}_b) < 5(p+1)/(p-2)^2$ for any $b=0,\ldots,p-1$.
\end{lemma}

Next, we show part (ii). The quantitative design $X = \tilde E_{b^*_2}$ is chosen as the leave-one-out design obtained from the Williams transformed design $E_{b^*_2}$. Applying Theorem~3 of \cite{wang2018optimal}, the $L_1$-distance {$d_1(\tilde E_{b_2^*})$ is at least $(p^2-7)/3 + (1/3)\sqrt{(p^2-1)/3}-(p-1)$, which simplifies to ${(m^2-m)}/{3}+\sqrt{m(m+2)/27}-2$} and this is of order $(1-\mathcal{O}(1/m))$ of the bound $d_{1,upper}$ in Lemma~\ref{dupper}. Following the result for the $L_1$-distance case and by the fact that
	$\left\{ \sum_{k=1}^{m} \vert x_{ik}-x_{jk}\vert ^2 \right\} \geq m^{-1}\left\{ \sum_{k=1}^{m} \vert x_{ik}-x_{jk}\vert  \right\}^2 $,
the $L_2$-distance $d_2(\tilde E_{b_2^*})$ is of order $\sqrt{{2}/{3}} (1-\mathcal{O}(1/m))$ of the bound $ d_{2,upper}$ in Lemma~\ref{dupper}.
\end{proof}

\begin{proof}[Proof of Proposition~\ref{prop:twob}]
	We only consider the case $b+b' =(p-1)/2$, the case $b+b'=(3p-1)/2 $ is similar.
	Condition $b+b' =(p-1)/2$ implies that both $b$ and $b'$ are smaller than $p/2$, and therefore $W(b) + W(b') = 2(b+b')=p-1$. Consider the $p\times m$ Williams transformed good lattice point designs $E_b$ and $E_{b'}$.  For any $i=1,\ldots,p-1$, the $i$th row in $E_{b}$ is $W( i h + b \text{ (mod } p)$, and the $(p-i)$th row in $E_{b'}$ is $W( (p-i) h + b'\text{ (mod } p)$, where $h = (1,\ldots,m)$. It follows from the construction of good lattice point design that the following two $2\times p$ matrices
	$$
	\left(
	\begin{array}{cc}
	i h + b \text{ (mod } p) & b \\
	(p-i) h +b' \text{ (mod } p) & b' \\
	\end{array}
	\right),
	$$
        $$
	\left(
	\begin{array}{cc}
	h  & 0 \\
	(p-1) h + (p-1)/2 \text{ (mod } p) & (p-1)/2 \\
	\end{array}
	\right)
	$$
	are the same up to column permutations. By Williams transformation, we have $W(x) + W(x') = p-1$ for any $x\in \mathcal{Z}_p$, where $x'= (p-1) x + (p-1)/2 \text{ (mod } p)$. Therefore,  the sum of the $i$th row in $E_{b}$ and the $(p-i)$th row in $E_{b'}$ is a constant row $p-1$, $i=1,\ldots,p-1$.
	For the leave-one-out designs $\tilde{E}_b$ and $\tilde{E}_{b'}$, their elements are changed from $E_b$ and $E_b'$ respectively by the maps $\mathcal{Z}_p \rightarrow \mathcal{Z}_m^+$:
    $$
	\pi_1(x) = \left\{ \begin{array}{ll}
	x+1, & x<W(b)  \\
	x, & x>W(b)
	\end{array} \right. ,
	$$
        $$
	\pi_2(x) = \left\{ \begin{array}{ll}
	x+1, & x<p-1-W(b)  \\
	x, & x>p-1-W(b)
	\end{array} \right..
    $$
	Between the two numbers $x$ and $p-1-x$, $x\in \mathcal{Z}_p$, there is exactly one that is smaller than $W(b)$ or $p-1-W(b) $, therefore it follows that the sum of the $i$th row in $\tilde E_{b}$ and the $(p-i)$th row in $\tilde E_{b'}$ is a row of constant $p$, which implies that $r_{ave} (\tilde  E_{b}) = r_{ave} (\tilde  E_{b'}) $.
\end{proof}

\begin{proof}[Proof of Lemma~\ref{lem:william_square}]
 For an even $m$, by the definition of the Williams transformation \eqref{wt}, the vector $h = (W^{-1}(0),W^{-1}(1), 
 \ldots, W^{-1}(m-1)) = (0, m-1, 1, m-2, \ldots, m/2-2, m/2+1, m/2-1, m/2)$. View $\mathcal{Z}_m=  \left\{0,\ldots,m-1\right\}$ as the ring of integers modulo $m$, we have the first order difference vector of $h$ is
 \begin{align}
 \label{eq:Wdif}
   & (W^{-1}(1) - W^{-1}(0),W^{-1}(2) - W^{-1}(1), \ldots, W^{-1}(m-1)  -W^{-1}(m-2))~(\textrm{mod } m)  \nonumber \\
   & =  (m-1, 2, m-3, 4, \ldots, m/2-2, 3, m/2, 1),
 \end{align}
 which is a permutation of $\{1,\ldots,m-1\}$.
 For each vector $h + i$ (mod $m$), $i = 1,\ldots,m-1$, the first-order difference vector is also \eqref{eq:Wdif}.
  Therefore, for the $m\times m$ matrix whose $i$th row is $h + (i-1)$ (mod $m$), $i = 1,\ldots,m$, each column is a permutation of $\mathcal{Z}_m$, and each pair of distinct numbers in $\mathcal{Z}_m$ appears exactly once.
  The matrix $O$ is obtained by replacing all the $0$ in the above matrix with $m$, hence it is a pair-balanced Latin square with entries from $\mathcal{Z}_m^+$.
\end{proof}

\begin{proof}[Proof of Theorem~\ref{theo:XO_other_m}]
  Part (i) follows directly by Lemma~\ref{lem:william_square}. For part (ii), when condition (a) holds, by Theorem ~1 of \cite{yin2022distance}, the pairwise $L_1$-distances between rows of $X$ all equal $m(m + 1)/3$, and $X$ is $L_1$-equidistant. When condition (b) or (c) hold, $d_1(X)/d_{1,upper} \ge 1-1/(m+1)$ by Theorem~4 of \cite{yin2022distance}. The conclusions under the $L_2$-distance follow by the same argument as the proof of Theorem~\ref{theo:XO_m} (ii).
\end{proof}

\begin{proof}[Proof of Lemma~\ref{lem:MCD}]
	If $O$ can be expressed as a row juxtaposition of $k$ Latin squares, let $\ell =(1,\ldots,km)^\T = (\ell_1^\T,\ldots,\ell_k^\T)^\T$ with $\ell_i = ((i-1)m+1, (i-1)m+2, \ldots, im)^\T$ for $i=1,\ldots,k$. For each $j=1,\ldots,m$, let the $j$th column of $X$ be $(\tilde{\ell}_{\pi(1)}^\T,\ldots,\tilde{\ell}_{\pi(k)}^\T)^\T$, where $(\pi(1),\ldots,\pi(k))$ is a random permutation of $\mathcal{Z}_{k}^+=\{1,\ldots,k\}$ and $\tilde{\ell}_{i}$ is obtained by randomly permuting the entries of $\ell_{i}$. Then, such $D=(X,O)$ has a marginally coupled structure.
	
 Conversely, if $D=(X,O)$ is a marginally coupled design, WLOG, we consider permuting the rows of $X$ to make its first column in ascending orders which becomes $\ell =(1,\ldots,km)^\T$. Under the same row-permutation in $O$, we obtain $O = \left(O_1^\T,\ldots,O_k^\T\right)^\T$, where each $O_i$ has $m$ rows. By collapsing the $n$ levels in $\mathcal{Z}_n^+$ to the $k$ levels in $\mathcal{Z}_k$ via mapping $x  \mapsto \lfloor (x-1)/m \rfloor$, the first column of $X$  becomes $(0,\ldots,0,1,\ldots,1,\ldots,k-1,\ldots,k-1)^\T$ with each level in $\mathcal{Z}_k$ repeated for $m$ times.
	By the definition of marginally coupled design $D=(X,O)$, for each level in each column of $O$, the corresponding rows in the first column of $X$ must be $(0,1,\ldots,k-1)^\T$ after level collapsing. This shows that each of the $m$ levels in $\mathcal{Z}_m^+$ appears exactly once in each column of $O_i$ where $i=1,\ldots,k$. Thus, $O_i$ is an $m\times m$ Latin square and $O$  can be expressed as a row juxtaposition of $k$ Latin squares $O_i$.
\end{proof}

\begin{proof}[{Proof of Theorem~\ref{theo:XOkm}}]
	(i). By Theorem~\ref{lem:balance}, all the designs $O_i$, $i=1,\ldots,k$, are pair-balanced with $t_{i,j} = 1$ for $1\leq i \neq j \leq m$. Thus, $O = \left(O_1^\T,\ldots,O_k^\T\right)^\T$ is pair-balanced with $t_{i,j} = k$.
	
	(ii). Let design $\tilde E = (\tilde{E}_{0}^\T, \ldots, \tilde{E}_{p-1}^\T )^\T$ be the sequence design when $k=p$, i.e., $k=m+1$. It is clear that $d_{H}(O) \geq d_{H}(\tilde E)$ for any sequence design $O$ with $k = 2, \ldots, p$. Thus, to show $d_{H}(O) \geq m-3$, it suffices to show $d_{H}(\tilde E) = m-3$.
	Let $\tilde e_i$ and $\tilde e_j$ be two distinct rows of $\tilde E$. When $\tilde e_i$ and $\tilde e_j$ are from the same $\tilde{E}_{s}$ for $s=0, \ldots p-1$, we have  $d_H(\tilde e_i, \tilde e_j) = m$, since $\tilde{E}_{s}$ is a Latin square.
	Next, we consider the case when $\tilde e_i$ and $\tilde e_j$ are from different $\tilde{E}_{s}$'s.
	Let the $p^2\times m$ matrix $E = ({E}_{0}^\T, \ldots, {E}_{p-1}^\T )^\T$ be the row juxtaposition of all Williams transformed designs $E_s$ for $s=0,\ldots,p-1$.  We can express $E = W(D)$ where $D=({D}_{0}^\T, \ldots, {D}_{p-1}^\T )^\T$ and $D_b=D_0+b=(x_{ik}+b)$ (mod $p$) for any $b\in \mathcal{Z}_p$. It is easy to see that $D_0$ is a difference matrix and $D$ is a $p^2$-run, $m$-factor and $p$-level orthogonal array of strength two with entries from  the Galois field $GF(p)$ \citep{hedayat1999}. Since each level combination occurs once in any two columns of $D$, the Hamming distance between any two rows of $D$ can only take two possible values: $m$ and $m-1$. Without loss of generality, suppose the two rows $\tilde e_i$ and $\tilde e_j$ in $O$ are the $i$th and $j$th rows in  $\tilde{E}_{a}$ and  $\tilde{E}_{b}$, respectively, where $0\leq a\neq b\leq p-1$. Then, the corresponding $i$th and $j$th rows of $E_a$ and $E_b$ are $e_i$ and $e_j$, respectively.
We have $e_i = W(i h +a$ (mod $p$)$)$ and $e_j = W(j h +b$ (mod $p$)$)$ for $h=(1,\ldots,p-1)$.
	By the construction of good lattice point sets, the two $2\times p$ matrices:
	$$
	\left(
	\begin{array}{cc}
	i h +a \text{ (mod } p\text{)} & a \\
	j h +b \text{ (mod } p\text{)} & b \\
	\end{array}
	\right)
	\text{ and }
	\left(
	\begin{array}{cc}
	h  & 0 \\
	j' h +b' \text{ (mod } p\text{)} & b' \\
	\end{array}
	\right)
	$$
	are the same up to column permutations, where $j'$ and $b'$ are elements in $\mathcal{Z}_p$ which are uniquely determined by $j=i j' \text{ (mod } p\text{)}$ and $b' = b-j' a \text{ (mod } p\text{)}$. Thus, after Williams transformation,
	$$
	\left(
	\begin{array}{cc}
	e_i  & W(a) \\
	e_j  & W(b) \\
	\end{array}
	\right)
	\text{ and }
	\left(
	\begin{array}{cc}
	W(h) & 0 \\
	W(j' h +b' \text{ (mod } p\text{)}) & W(b') \\
	\end{array}
	\right)
	$$
	are the same up to column permutations.
	Since $d_H( e_i, e_j)=m$ or $m-1$, the two rows in
	$$
	\left(
	\begin{array}{cc}
	W(h) & 0 \\
	W(j' h +b' \text{ (mod } p\text{)}) & W(b') \\
	\end{array}
	\right)
	$$
	have at most one column with the same elements. Next, we consider the two rows  $\tilde e_i$ and $\tilde e_j$ whose elements are mapped from $e_i$ and $e_j$ via mapping $\mathcal{Z}_p \rightarrow \mathcal{Z}_m^+$:
	$$\pi_1(x) = \left\{\begin{array}{ll}
	x+1, & x<W(a)  \\
	x, & x>W(a)
	\end{array} \right.
	$$
 $$
	\pi_2(x) = \left\{\begin{array}{ll}
	x+1, & x<W(b)  \\
	x, & x>W(b)
	\end{array} \right. .$$
	It suffices to show $d_H(\tilde e_i, \tilde e_j) \geq d_H(e_i,  e_j)-2$.
	Suppose $x \neq j'x+b' \text{ (mod } p\text{)}$ and $[W(x), W(j'x+b' \text{ (mod } p\text{)} )]^\T$ is a column of $(e_i^\T, e_j^\T)^\T$ with $\vert W(x) - W(j'x+b' \text{ (mod } p\text{)} ) \vert=1$.
	Then, after the above mapping, $\pi_1( W(x) ) =  \pi_2[W(j'x+b' \text{ (mod } p\text{)} )]$ if and only if either of the following cases holds:\\
	Case 1:
	$W(b)-1 < W(x) < W(a) $ and $W(x) = W(j'x+b' \text{ (mod } p\text{)} )-1$;\\
	Case 2:  $W(a) < W(x) < W(b)+1 $ and $W(x) = W(j'x+b' \text{ (mod } p\text{)} )+1$.\\
	For simplicity, we only show the proof of $d_H(\tilde e_i, \tilde e_j) \geq d_H(e_i,  e_j)-2$ for Case~1, and its proof for Case~2 is very similar. There are two situations under Case 1. \\
	Case 1(a): $j' \neq p-1$. There are at most two $x$'s in $\mathcal{Z}_p$ satisfying $W(x) = W(j'x+b' \text{ (mod } p\text{)} )-1$ which correspond to
	$x + (j'x+b' \text{ (mod } p\text{)} ) =p$ for $x>(p-1)/2 $ and $x + (j'x+b' \text{ (mod } p\text{)} ) =p-1$ for $x\leq(p-1)/2 $, respectively. Thus, it is clear that $d_H(\tilde e_i,\tilde e_j) \geq d_H( e_i, e_j) -2 \geq m-3$ under this situation. \\
	Case 1(b): $j' = p-1$. Then, we have $j=p-i$.
	The condition $W(x) = W(j'x+b' \text{ (mod } p\text{)} )-1$ holds only if $x + (j'x+b' \text{ (mod } p\text{)} ) =p $ or $x + (j'x+b' \text{ (mod } p\text{)} ) =p-1 $.
	If $b'\neq 0$ and $b'\neq p-1$, no $x \in \mathcal{Z}_p$ exists such that $\vert W(x) - W(j'x+b' \text{ (mod } p\text{)} ) \vert=1$. Therefore, $d_H(\tilde e_i,\tilde e_j) = d_H( e_i, e_j)\geq m-1$.
	If $b'= 0$, we have $b=p-a$ and $|W(a) - W(b) |=1$. Similarly, if $b'= p-1$, we have $b=p-1-a$ and $\vert W(a) - W(b) \vert=1$ because $a\neq b$. In both scenarios, $W(b)-1< W(a) $ holds if and only if $W(b) = W(a)-1$. Therefore, $W(b)-1 < W(x) < W(a) $ if and only if $x=b$ and $d_H(\tilde e_i,\tilde e_j) = d_H( e_i, e_j)-1\geq m-2$.
	
	(iii). For $1\leq u\neq v \leq m$, it is easy to verify that $r_{uv}(O) = \sum_{i=1}^{k} r_{uv}(O_i) /k$. Therefore, $r_{ave}(O) \leq \sum_{i=1}^{k} r_{ave}(O_i) /k$. By Lemma~\ref{lem:r}, all the leave-one-out Williams transformed designs $O_i$ have $r_{ave}(O_i) < \mathcal{O}(1/m)$ for $i=1,\ldots,k$. Thus, design $O$ has $r_{ave}(O) < \mathcal{O}(1/m) \rightarrow 0$ as $m\rightarrow \infty$.
	
	(iv). For each $X_i$, $i=1,\ldots,k$, the 3-step construction method gives $d_q(X_i) = d_q(\tilde E_{b^*})$ for $q=1,2$. Thus, we have $d_q(X) \leq d_q(\tilde E_{b^*}) $. For any $1\leq i\neq i'\leq k$, let $x^{(i)}_{\ell}$ and $x^{(i)}_{\ell'}$ be the $\ell$th and $\ell'$th rows in design $X_i$ and $X_i'$, respectively. Then, the construction method of $X$ gives
 \begin{align*}
     & d_q( x^{(i)}_{\ell}, x^{(i')}_{\ell'}) \geq d_q( x_{\ell},  x_{\ell'}+m  ) > \Big\{ \sum_{j=1}^{m}(m+1-j)^q \Big\}^{1/q}   > d_q(\tilde E_{b^*}),
 \end{align*}
	where $x_{\ell}$ is the $\ell$th row in $\tilde E_{b^*}$ and $q=1,2$. The conclusion follows by Theorem~\ref{theo:XO_m}~(ii).
		
	(v). For any $j=1,\ldots,m$,  collapsing each level in  $m \ell_j \otimes 1_m +  f_j $, i.e. the $j$th column of $X$, through mapping $x  \mapsto \lfloor (x-1)/m \rfloor$ yields a column $\ell_j \otimes   1_m$. Therefore, for each level in each column of $O$, the corresponding rows in the collapsed column $\ell_j \otimes 1_m$ must be $(0,1,\ldots,k-1)^\T$, which means $D=(X,O)$ must be a marginally coupled design.
\end{proof}

\begin{proof}[Proof of Lemma~\ref{lemma:MCD2}]
	Let $g(\cdot)$ denote the map $g: \mathcal{Z}_n^+ \rightarrow \{0,1\}: x  \mapsto \lfloor (x-1)/m \rfloor$ that collapses the $n$ levels in $\mathcal{Z}_n^+$ to two levels $\{0,1\}$.
	By Lemma~\ref{lem:MCD} and its proof, we can always permute the rows of $D=(X,O)$ such that
	$$
	D = (X,O) = \left(
	\begin{array}{cc}
		X_1 & O_1 \\
		X_2 & O_2
	\end{array}
	\right),
	$$
	where each $O_i$ an $m\times m$ Latin square having the first column $(1,\ldots,m)^\T$, and $g(X)$ has the first column $( 0_m^\T,  1_m^\T)^\T$ after the collapsing map $g$, where $  0_m$ is the vector  of $m$ zeros.
	
	For any $i=1,2$ and $j\in \mathcal{Z}_m^+$, let $x^{(i)}_{j} = (x^{(i)}_{1j},\ldots, x^{(i)}_{mj})^\T$ denote the $j$th column of $X_i$. We have $g(x^{(1)}_{1}) =  0_m$ and $g(x^{(2)}_{1}) =  1_m$.
	We claim that for any $i=1,2$,  $j=2,\ldots,m$, by the collapsing map $g$,  $g(x^{(i)}_{j})$ must become a constant column $ 0_m$ or $ 1_m$. If this is not true, WLOG, suppose that for the second column, $g(x^{(1)}_{2}) = ( 0_a^\T,  1_{m-a}^\T)$, where $0<a<m$. Then because the first columns of $O_1$ and $O_2$ are both $(1,\ldots,m)^\T$,
if the structure is marginally coupled that $g(x^{(2)}_{2}) = ( 1_a^\T,  0_{m-a}^\T)$. Furthermore, the first $a$ rows of $O_1$ (respectively, $O_2$) and the last $m-a$ rows of $O_2$ (respectively, $O_1$) form an $m\times m$ Latin square.  Let $o^{(i)}_{\ell} = (o^{(i)}_{\ell1}, \ldots, o^{(i)}_{\ell m})$ denote the $\ell$th row of the Latin square $O_i$, $1\leq \ell \leq m$, $i=1,2$. Consider the $a^2$ Hamming distances $d_H(o^{(1)}_{\ell},  o^{(2)}_{\ell'}  ) $ where $1\leq \ell,\ell' \leq a$.
	We have
	$$ \sum_{\ell=1}^{a} \sum_{\ell'=1}^{a} \sum_{j=1}^{m} d_H(o^{(1)}_{\ell},  o^{(2)}_{\ell'}  ) = \sum_{j=1}^{m}  \sum_{\ell=1}^{a} \sum_{\ell'=1}^{a} I (o^{(1)}_{\ell},o^{(2)}_{\ell'}) = ma(a-1).$$
	Therefore we have the average of these $a^2$ Hamming distances is $(1-1/a)m$ and $d_H(O) \leq \lfloor (1-1/a)m \rfloor$. Similarly, we have $d_H(O) \leq \lfloor (1-1/(m-a))m \rfloor$. Now we have $m-2 \leq  d_H(O) \leq \min\{ \lfloor (1-1/a)m \rfloor,  \lfloor (1-1/(m-a))m \rfloor\}$. The equality holds if and only if $m$ is even and $a=m-a=m/2$.
	However, this is a contradiction because the number of pairs of runs in $O$  with $d_H(\alpha_i,\alpha_j)=m-2$ is smaller than $m^2/2$, which proves that for any $i=1,2$,  $j=2,\ldots,m$, $g(x^{(i)}_{j})=  0_m$ or $ 1_m$.
	
	By now we have $d_q(X)\leq d_q(X_1) = d_q( X_1 - m \cdot g(X_1) )$. The conclusion then follows from the lemma~\ref{dupper} and the fact that the matrix $(X_1 - m \cdot  g(X_1))$ is an LHD.
\end{proof}

\begin{proof}[Proof of Proposition~\ref{prop:MCD2}]
	(i)  For two runs in the same Latin square $O_1$ or $O_2$, their Hamming distance is $m$. Consider two runs from $O_1$ and $O_2$, respectively, suppose $\tilde e_i$  is the $i$th run of $O_1=\tilde{E}_{b_1^*}$ and $\tilde e_j$ is the $j$th row of $O_2=\tilde{E}_{p-b_1^*}$. The corresponding  $i$th and $j$th rows of $E_{b_1^*}$ and $E_{p-b_1^*}$ are $e_i = W(i h +b_1^*$ (mod $p$)) and $e_j = W(j h -b_1^*$ (mod $p$)) where $h=(1,2,\ldots,p-1)$.  From the proof of Theorem~\ref{theo:XOkm} (ii) we have
	$$
	\left(
	\begin{array}{cc}
	e_i  & W(b_1^*) \\
	e_j  & W(p-b_1^*) \\
	\end{array}
	\right)
	\text{ and }
	\left(
	\begin{array}{cc}
	W(h) & 0 \\
	W(j' h +b'  \text{ (mod } p)) & W(b') \\
	\end{array}
	\right)
	$$
	are the same up to column permutations,  where $j'$ and $b'$ are the elements in $\mathcal{Z}_p$ uniquely determined by $j=i j' \text{ (mod } p)$, $b' = p-(j'+1) b_1^* \text{ (mod } p)$.
	The two rows  $\tilde e_i$ and $\tilde e_j$ have elements changed from $ e_i$ and $ e_j$ respectively  by the maps $\mathcal{Z}_p \rightarrow \mathcal{Z}_m^+$:
	$$\pi_1(x) = \left\{ \begin{array}{ll}
	x+1, & x<W(b_1^*)  \\
	x, & x>W(b_1^*)
	\end{array} \right. ,
	$$
 $$
	\pi_2(x) = \left\{ \begin{array}{ll}
	x+1, & x<W(p-b_1^*)   \\
	x, & x>W(p-b_1^*)
	\end{array} \right..$$
Without loss of generality, suppose $b_1^*<(p-1)/2$, then $W(b_1^*)=2b_1^*$ and $W(p-b_1^*) = 2 b_1^* -1$.
	Following the proof of Theorem~\ref{theo:XOkm} (ii), Case 1: $ W(p-b_1^*)-1 < W(x) < W(b_1^*) $ and $W(x) = W(j' x +b' \text{ (mod } p)) -1$ hold if and only if $x = p-b_1^*$, $j'=-1 $ and $b'=0$, which is equivalent to $x = p-b_1^*$ and $j=p-i $.
	For these pairs $(i, j=p-i)$, $i\in \mathcal{Z}_m^+$, we have that each $(e_i^\T, e_j^\T)^\T$ has exactly one column, being $(0,0)^\T$ and the other columns have different elements. Therefore, $d_H(e_i, e_j) = m-1$, and we have further $d_H(\tilde e_i,\tilde e_j)=m-2$. For the other pairs $(i, j)$, $j\neq p-i$ and $i,j\in \mathcal{Z}_m^+$, $d_H(\tilde e_i,\tilde e_j)=m-1$ or $m$. The conclusion follows.
	
	(ii) From the proof of (i), we have $ d_H(O) = m-2$ and there are exactly $m$ pairs of runs in $O$ with the Hamming distance $m-2$. Thus, the desired conclusion follows from Theorem~\ref{theo:XOkm} (iv) and Lemma~\ref{lemma:MCD2}.
\end{proof}

\end{appendices}

\end{document}